\documentclass[10pt,journal,a4paper]{IEEEtran}

\usepackage{cite}
\usepackage{xcolor}
\usepackage{graphicx}
\usepackage[latin1]{inputenc}
\usepackage[T1]{fontenc}
\usepackage{amsmath,amsfonts,amsbsy,amssymb}
\usepackage{mathabx}
\usepackage{url}

\usepackage{mathrsfs}
\usepackage[nolist]{acronym}
\usepackage{tabularx}
\usepackage{amssymb}
\usepackage{amsmath}
\usepackage{graphicx}
\usepackage[latin1]{inputenc}
\usepackage{graphicx}
\usepackage{cite}
\usepackage{multirow}
\usepackage{wasysym}
\usepackage{multirow}
\usepackage{float}
\usepackage{color}
\usepackage{enumitem}
\usepackage{amsthm}	
\usepackage{xcolor,soul}
\definecolor{aqua}{rgb}{0.0, 1.0, 1.0}
\definecolor{laranja}{rgb}{1.0, 0.6, 0.2}

\newtheorem{proposition}{Proposition}

\newtheorem{corollary}{Corollary}

\newcommand{\expectation}[2]{\operatorname{E}_{#2}\left[ #1 \right]}

\begin{document}

\title{Finite Blocklength Communications in Smart Grids for Dynamic Spectrum Access and Locally Licensed Scenarios}
\author{
	\IEEEauthorblockN{Iran Ramezanipour, Parisa Nouri, Hirley Alves, Pedro J. H. Nardelli, Richard Demo Souza and Ari Pouttu} 

	\thanks{I. Ramezanipour, P. Nouri, H. Alves and A. Pouttu are with the Centre for Wireless Communications (CWC), University of Oulu, Finland. Contact: firstname.lastname@oulu.fi.}
	\thanks{P. Nardelli is with Lappeenranta University of Technology, Lappeenranta, Finland. Contact: pedro.nardelli@lut.fi.}
	\thanks{R.D. Souza is with Federal University of Santa Catarina (UFSC), Florianopolis, Brazil. Contact: richard.demo@ufsc.br.}
	\thanks{This work is partially supported by Aka Project SAFE (Grant n.303532), Strategic Research Council/Aka BCDC Energy (Grant n.$292854$), Finnish Funding Agency for Technology and Innovation (Tekes), Bittium Wireless, Keysight Technologies Finland, Kyynel, MediaTek Wireless, Nokia Solutions and Networks and CNPq, Brazil.}
}
%
\maketitle

\begin{abstract}

This work focuses on the performance analysis of short blocklength communication with application in smart grids. 
We use stochastic geometry to compute in closed form the success probability of a typical message transmission as a function of its size (i.e. blocklength), the number of information bits and the density of interferers. Two different scenarios are investigated: (\textit{i}) dynamic spectrum access where the licensed and unlicensed users, share the uplink channel frequency band and (\textit{ii}) local licensing approach using the so called micro operator, which holds an exclusive license of its own.
Approximated  outage probability expression is derived for the dynamic spectrum access scenario, while a closed-form solution is attained for the micro-operator.
The analysis also incorporates the use of retransmissions when messages are detected in error.
Our numerical results show how reliability and delay are related in either scenarios.
\end{abstract}

\begin{IEEEkeywords}
	Machine-to-machine, micro-operator, short message, dynamic spectrum access, local licensing, ultra-reliability, Poisson Point Process 
\end{IEEEkeywords}

\section{Introduction}


Wireless networks have become an indispensable part of our daily life through a wide range of applications.
For instance, we may now remotely monitor and control different processes within our homes, workplaces or even at industrial environments. 
In the upcoming years the advances in wireless communications shall be even more seamless and will provide connectivity through the so called Internet of Things (IoT) \cite{ART:Ericsson2015}.
It is envisioned that by the year 2020, billions of devices (including sensors and actuators) will be connected to the Internet, gathering all kinds of data and generating  a huge economic impact \cite{ART:McKinsey2015, ART:Perera-A2014}. 

One of the {key enablers of this future is the so-called machine-type communication} (MTC) where a large number of devices  will perform sensing, monitoring, actuation and control tasks with minimal or even no human intervention \cite{ART:Ericsson2015}.
In other words, MTC -- also known as machine-to-machine (M2M) communication -- incorporates sensors, appliances and vehicles and this is expected to lead to a decrease of human-centric connections~\cite{ali2015next}. 
MTC is also one of the cornerstones of the upcoming 5G communication technologies.
As discussed before, it contemplates a massive deployment of devices communicating with diverse range of requirements in terms of reliability, latency and data rates \cite{tullberg2016metis, bockelmann2016massive, hu, onel2} .
For example, \cite[Table. 2]{dawy2015towards} lists the main requirements and features for different use cases of MTC over cellular networks.

Intelligent transportation systems, industrial automation and smart grids have already been deployed using the IoT concept, which is also one of the main driving technologies of 5G ~\cite{schulz2017latency,8067687}.
It should be noted that the mentioned applications have different reliability requirements.
For instance, some smart grid applications requires a reliability as low as $10^{-6}$, which in their turn is less strict if compared to some industrial automation applications~\cite{schulz2017latency}.  

\vspace{-2ex}
\subsection{Dynamic Spectrum Access and Locally Licensed in 5G} 

One of the main features of 5G will be its capability of connecting very large number of devices in different locations, while being able to serve case specific needs of different applications. 
Indoor networks are responsible for the larger part of the mobile traffic, hence, it is essential to build new, more efficient indoor small cell networks.
This will require more spectrum, which makes its availability a big challenge to tackle. 
There are generally two ways to access the available spectrum in a network which are \cite{del2016d1}: (\textit{i}) Individual Authorization (Licensed), and (\textit{ii}) General Authorization (Unlicensed). 
There are five different allocations scenarios associated with the previously mentioned access schemes, namely dedicated licensed spectrum, limited spectrum pool, mutual renting, vertical sharing and unlicensed horizontal sharing.

Cognitive radio has gained a high popularity during the past few years since it makes a more efficient use of the frequency spectrum possible \cite{1}, \cite{Nardelli2015}. 
Dynamic spectrum access is one of the many interesting aspects of cognitive radios \cite{peha2009sharing} where the unlicensed-users can use the same frequency band as the licensed-users while not affecting their transmission. For that sake, the unlicensed-users evaluate the spectrum usage and then transmit if the channel is free, otherwise, they postpone their transmission or use other frequency bands \cite{saleem2014primary}, \cite{Nardelli2015}. Dynamic spectrum access can also be implemented in smart grid communication networks. In \cite{yu2011cognitive,gungor2012cognitive}, the use of dynamic spectrum access in smart grid communication networks is evaluated and its suitability is positively assessed.  

Another interesting approach towards reaching spectral efficiency is using the non-orthogonal cognitive radio techniques. Despite being relatively new, many valuable research works have been done in this field. In \cite{op}, authors develop a cognitive radio scheme for multicarrier wireless sensor networks by studying a dense wireless sensor network model where the sensors can opportunistically use the primary users' spectrum for their transmissions. Unlike the previous case, this model does not require the sensor nodes to sense the channel before transmission which is useful in terms of maintaining the limited resources of the sensors. Authors in \cite{qiu2012interference} propose a number of new interference control and power allocation methods for cognitive radios which sheds light on the primary and secondary users' power allocation requirements. An interesting work has been done in \cite{darsena2016convolutive} where a new spectrum sharing model is proposed for multicarrier cognitive radio systems in which the secondary users can simultaneously use the primary users' frequency band for their transmission while actually improving the primary users' transmission by the mean of convolutive superposition. 

Moreover, in \cite{han2008cooperative} and \cite{han2009cooperative}, two different spectrum access protocols are proposed for the secondary networks via controlled amplify-and-forward relaying and cooperative decode-and-forward relaying. These protocols are proven to not have a negative impact on the rate and outage probability of the primary network. Another interesting work while using amplify-and-forward scheme has also been done in \cite{verde2015amplify} where the proposed model makes it possible for secondary users to use the primary users' frequency channel for their transmissions even when the primary network is active. By using this model, it is possible to improve the secondary user's packet delay and primary users' achievable rate. A useful model for improving the secondary network's achievable rate is also introduced in \cite{shin2011time} where the authors achieve this goal by applying superposition coding to a collaborative spectrum sharing scenario.    

While all the aforementioned spectrum sharing and dynamic spectrum access models shall be a part of the future wireless communications, having exclusively licensed spectrum (locally licensed) is crucial for 5G to be able to meet some Quality of Service (QoS) requirements \cite{del2016d1}. 
For instance, the micro-operator ($\mu\mathrm{O}$) concept has been introduced as a mean for local service delivery in 5G which will benefit from having exclusively licensed spectrum. 
$\mu\mathrm{O}$s make the previously mentioned case-specific services in the future indoor small cell networks possible~\cite{boost},\cite{del2016d1}.
Micro-operators have their own specific infrastructure which enables them to handle different kinds of Mobile Network Operator (MNO) users while also collaborating with the network infrastructure vendors, facility users/owners, utility service companies and regulators.

$\mu\mathrm{O}$s shall help the MNOs specially in the areas that the traffic demand is high by offering them indoor capacity. 
It should be noted that the functionality of a $\mu\mathrm{O}$ depends on the available spectrum resources, which are limited.
As mentioned earlier, having exclusive licensed spectrum is very important for the success of 5G so the $\mu\mathrm{O}$s are the entities in a network that can benefit from it since regulators are able to issue local spectrum licenses for their own usage within a specific location \cite{del2016d1,boost}.
While we acknowledge all the valuable works mentioned earlier that have been done in the area of spectrum sharing which are also relevant for smart grid applications, it is important to mention that, in this study, our focus is to analyze the performance of the two dynamic spectrum access and locally licensed models.
To do so, next we review the reliability requirements for smart grid applications and recent works in the MTC area.

\vspace{-2ex}
\subsection{Reliability in Smart Grids and the Role of MTC} 

Communication systems have been traditionally studied using the notion of channel capacity that assumes very large (infinite) blocklength \cite{Polyanskiy2010a, chung2001design, shannon1961two}; this is a reasonable benchmark for practical systems with blocklength in the order of thousands of bits. 
MTC, however, often uses short messages which is not currently supported by the wireless networks and periodic data traffic, coming from a massive number of devices. The same assumptions in terms of channel capacity cannot be directly applied to short blocklength messages as pointed out in \cite{Polyanskiy2010a, ART:Durisi-PIEEE2016}. 
This imposes the need for a new paradigm on the network design and analysis architecture to support such amount of connected devices with their heterogeneous requirements \cite{ART:Perera-A2014}.
%
%

New information theoretic results have been presented to evaluate the performance of short blocklength systems, from point-to-point additive white Gaussian noise (AWGN) links up to Multiple-Input Multiple-Output (MIMO) fading channels \cite{ART:Durisi-PIEEE2016}.
In \cite{Makki2014a, Makki2015}, the authors investigate retransmission methods and interference networks in this context.
Nevertheless, network level analysis is still missing in the literature, except for \cite{Makki2016}, where the authors utilize Poisson Point Process (PPP) to characterize the network deployment and interference, considering finite block codes in a cellular network context. 

Smart grids characterize the modernization process that the power grids undergo and is one important case for MTC~\cite{ni,Nardelli2014}. Communication systems are one the most important parts of the smart grids and different communication technologies are currently being used in smart grids most of which use the existing communication technologies such as PLC, fiber optical communications and LTE \cite{piti2017role}. However, considering the rapid advancements of the communication systems towards 5G, smart grid communication systems should also be designed in a way that would be compatible with the newest telecommunication technology requirements which would not be really possible by using the traditional communication systems anymore. Therefore, smart grids are an interesting topic to be studied under the umbrella of 5G, especially considering that the requirements imposed by smart grids have hitherto overlooked specially with respect to massive connectivity and ultra reliable low latency communication \cite{8067687}. Hence, motivated by smart grids stringent requirements \cite{kuzlu2014communication,8067687}, here we focus on two different scenarios looking at the ultra-reliability using finite blocklength (short messages) in order to reduce latency and capture practical aspects regarding the message size, which is one of the novel aspects of 5G.
The reliability requirement of a smart grid network varies from one application to another \cite{kuzlu2014communication}. 
For instance, applications such as remote meter reading have less strict reliability requirements ($98\%$) while high-voltage grids require high reliability (more than $99.9\%$) in addition to low latency. Moreover, applications such as teleprotection in smart grid networks also require very high reliable data transmission between the power grid  substations within a very short period of time, in the order of few milliseconds. Smart girds also may need to have real time monitoring and control and should be able to react immediately to the changes in the network which means there is going to be a need for ultra reliable communications with $99.999\%-99.99999\%$ reliability level and a low latency, around $0.5-8 ms$ \cite{osseiran2014scenarios, fallgren2013scenarios, 8067687}. In this paper, we show that it is possible to achieve these different and strict requirements with our proposed models using finite blocklength communications which require a completely different design settings compared to what is currently being used in for instance, LTE or WiFi \cite{ART:Durisi-PIEEE2016} . 

\subsection{Related Work}
Ultra reliable communications and finite blocklength have become popular topics and many studies have been done in this field, however, there are still many issues that need to be addressed. For instance, \cite{Polyanskiy2010a}, \cite{Durisi2016,Yang2014c}, being amongst the first fundamental works in the finite blocklength area, where they set the foundation of the finite blocklength communications for cases such as block fading, MIMO and AWGN, also open up a variety of topics that need to be tackled. In \cite{ART:Durisi-PIEEE2016}, the necessity of studying and employing short packet communications is explained and is foreseen as one of the main enablers of the future telecommunication technologies. The authors bring into light the recent achievements in the field of short message transmission while also emphasizing the need for more research to be done on several open challenges. Valuable works have been done in \cite{Makki2014a,Makki2015} where authors use the finite blocklength notion to analyze the throughput of different wireless networks. In \cite{Makki2016}, a model similar to ours is investigated where they also use PPP to characterize the cellular networks and evaluate the outage and throughput of the network in which a base station is connected to its nearest neighbor. However, this is not the case in our model. In the models studied in this paper, we use PPP where users are at a fixed distance, we use a different characterization of the signal-to-interference-plus-noise ratio (SINR) distribution and constraints which have led to totally different analysis and results. Our focus is not on a single link communication, hence, the SINR in this case captures interference, and to some extend, the network dynamics as well. We provide a general approximation for the outage probability which does not rely on any specific distribution of the SINR. Here, we focus on massive connectivity constrained by reliability which are imposed on the network from the application at hand (smart grids). We show that it is possible to achieve reliable and ultra reliable communication using finite blocklength which is also a characteristic of smart metering transmissions. It is shown how important it is to know how reliability and latency are affected by the increasing number of interferers. We also propose two different schemes in order to overcome this problem, namely dynamic spectrum access and local licensing scenarios, that can be used based on the restrictions imposed by the application.

\vspace{-2ex}
\subsection{Contributions}

The followings are the main contributions of this paper.
\begin{itemize}
	\item The general expression of the outage probability as a function of the number of information bits, blocklength and density of interferers in closed-form in the finite blocklength regime.
	
	\item Closed-form approximations of the outage probability are derived for both the dynamic spectrum access and  locally licensed scenarios, under different conditions in the finite blocklength regime.
	
	\item Two different schemes are evaluated to be used in different smart grid applications based on network density and reliability requirements, considering the finite blocklength regime. We show that these schemes can not be used in every network model, since they have different requirements, while we also show when it is most suitable to use which.
	
	\item A general expression for the delay is proposed where the effect of retransmissions is investigated. 
\end{itemize}

Table \ref{tbl:1} summarizes the functions and symbols that are used in this paper. The rest of the paper is divided as follows.
	Section \ref{sect:smodel} introduces the network model with how the communication model using the short blocklength is modeled, while Section \ref{sect:tput_analysis} details the outage analysis for both of the proposed scenarios.
	Section \ref{inv} presents the numerical results and how the two scenarios can reach the reliability requirements of smart grids, and how retransmissions affect the reliability and latency. Section \ref{sect:conclusion} concludes this paper.                                

\begin{table}
	\caption{Summary of the functions and symbols.}
	\centering
	\begin{tabular}{ | l | l | l }
		\hline
		Symbol & Expression  \\ \hline
		$F_{Z}\left[\cdot\right]$ & SINR CDF  \\ \hline
		$f_{Z}\left[\cdot\right]$ & SINR PDF \\ \hline
		$\operatorname{E}_T\left[\cdot\right]$ & Expectation \\ \hline
		$\Gamma\left[\cdot\right]$ & Gamma Function \\ \hline
		$\operatorname{Q}\left[\cdot\right]$ & Gaussian Q-Function\\ \hline
		$V(z)\left[\cdot\right]$ & Channel Dispersion \\ \hline
		$T_s$ & Symbol Time \\ \hline
		$\hat{\Phi}$ & Poisson Point Process\\ \hline
		$\mathbf{x}$ & Set of Interferers' Locations \\ \hline
		$\mathbf{H}$ & Channel Fading Coefficient \\ \hline
		$\lambda$ & Network Density \\ \hline
		$|x_i|$ & Distance Between Node $x_i$ and the Reference Receiver\\ \hline 
		$h_0$ & Channel Fading Coefficient in the Reference Link \\ \hline
		${W}_{p}$ & Licensed User Transmit Power\\ \hline
		${W}_{s}$ & Unlicensed User Transmit Power\\ \hline
		$\alpha$ & Path Loss Exponent \\ \hline
		$Z$ & SINR\\ \hline
		$I$ & AWGN Power \\ \hline
		$\xi$ & Noise Level \\ \hline
		$k$ & Information Bits \\ \hline
		$n$ & Blocklength \\ \hline
		$R$ & Coding Rate \\ \hline
		$d$ & Distance Between Transmitter and Receiver \\ \hline
		$m$ & Number of Transmission Attempts\\ \hline
		$\gamma_{th}$ & SINR Threshold\\ \hline
		$\nu$ & Number of Channel Uses for ACK/NACK \\
         \hline
	\end{tabular}

     \label{tbl:1}
    \end{table}


\noindent{\textbf{Notation:}} 
The probability density function (PDF) and the cumulative distribution function (CDF) of a random variable (RV) $T$ are denoted as $f_T(t)$ and $F_T(t)$, respectively, while its expectation is $\operatorname{E}_T\left[\cdot\right]$. The gamma function is defined as $\Gamma(t)$\cite[Ch 6, \S6.1.1]{BOOK:ABRAMOWITZ-DOVER03}, and the regularized upper incomplete gamma function is denoted as $\Gamma(s,t)$ \cite[Ch 6, \S6.5.3]{BOOK:ABRAMOWITZ-DOVER03}. $\operatorname{Q}(\cdot)$ denotes the Gaussian $Q$-function $\operatorname{Q}(t)=\tfrac{1}{\sqrt{2 \pi}}\int_{t}^{\infty} \exp\left(-\frac{u^2}{2} \right)\mathrm{d}u = \tfrac{1}{2} \operatorname{Erfc}\left(\tfrac{t}{\sqrt{2}}\right)$ \cite[\S.26.2.3]{BOOK:ABRAMOWITZ-DOVER03}.

\section{Network Model}\label{sect:smodel}

 The conventional methods for evaluating communication networks are not usually a suitable choice when studying large wireless networks due to several reasons such as focusing on signal to noise ratio (SNR) rather than signal to interference plus noise ratio (SINR) or the fact that the interference in these kinds of networks depends on the path loss, meaning that it also depends on the network geometry. However, using stochastic geometry and Poisson point process to model large wireless networks have proven to be a useful tool in solving the challenges faced by the classical methods \cite{haenggi2012stochastic}. Hence, in this paper, we assume a dense network where the position of the interferers is modeled as a Poisson point process~\cite{Makki2016}.
Formally, we are dealing with a Poisson field of interferers~\cite{cardieri2010modeling} where the distribution of nodes that are causing interference follows a $2$-dimensional Poisson point process $\hat{\Phi}$ with density $\lambda>0$ (average number of nodes per $m^2$) \cite{haenggi2012stochastic} over an infinite plane \cite{Nardelli2015} \cite{de2016contention}.
This process is represented as $\hat{\Phi} = (X, H)$, where 
$X$ is the set of interferers' locations and $H$ represents the set of quasi-static channel fading coefficients in relation to the reference receiver located arbitrarily at the origin\cite{Nardelli2015,haenggi2012stochastic}. Notice that from Slivnyak theorem \cite{haenggi2012stochastic} an arbitrarily-located receiver is placed at the center of the Euclidean space and is used as a fixed point of reference which makes the estimation of the position of the surrounding elements possible \cite{haenggi2012stochastic}.
Note that $x_i \in X$ denotes a position in the 2-dimensional plane and  $i\in \mathbb{N}^+$.
Besides, $h_{i} \in H$ is assumed to be constant during the transmission of one block, which takes $n$ channel uses, and during a spatial realization of the point process. 
We assume the fading coefficients follow a Rayleigh distribution, so that $h_{i}^2$ is exponentially distributed $h_{i}^2 \sim \mathrm{Exp}(1)$.
The fading coefficient $h_0$ is associated with the reference link, composed by a transmitter located at distance $d$ from its associated receiver. It should be noted that since here we are using unbounded path loss model, $\alpha > 2$ \cite [Ch. 5]{haenggi2012stochastic}.

\subsection{Communication Model}
Signal propagation is modeled using large-scale distance-based path-loss and Rayleigh fading. 
The received power at the reference receiver from the interferer $i$ is given by
\begin{align}\label{eq:rxPower}
{I}_{i} &= W_p \, h_{i}^2 \, |x_i|^{-\alpha},
\end{align}
where $W_p$ is the transmit power, $\alpha > 2$ is the path-loss exponent and $|x_i|$ is the distance between the node $x_i$ and the reference receiver.
It is important to note that $W_p$ is related to the interferers' transmit power; the same equation is valid for the reference link $0$, which may have a different transmit power denoted by $W_s$.
Then: ${I}_{0} = W_s h_{0}^2\, d^{-\alpha}$.
The SINR\cite{wildman2014joint} is defined as the random variable $Z \triangleq \tfrac{W_s h_0 d^{-\alpha}}{I + \eta}$, and $I= \sum_{i\in \mathbb{N}^+}{I}_{i}$ and $\eta>0$ is the AWGN power.
Under these conditions the SINR cumulative distribution function (CDF)\footnote{Notice that we denote $F_Z (z) = F_Z( z | \alpha, \lambda,\zeta,\xi)$ and $f_Z (z) = f_Z( z | \alpha, \lambda,\zeta,\xi)$, except when the parameters are being manipulated, then we explicitly indicate them.}  is given as \cite[Corollary 1]{wildman2014joint}
\begin{equation}
\label{eq:cdfSINR}
F_Z( z | \alpha, \lambda, \zeta, \xi ) \!=\!  1 \!-\! \Pr\left[Z > \!z \right]\!=\! 1\!-\! \exp\!\left(\! -\! \zeta \lambda z^{\frac{2}{\alpha}} \!-\! \xi z \!  \right),
\end{equation}
where $\zeta\!=\!\kappa \pi d^2(\frac{W_p}{W_s})^\frac{2}{\alpha}$, $\xi\!=\!\frac{\eta d^\alpha}{W_\mathrm{s}}$, and $\kappa\!=\! \Gamma\left(1+\frac{2}{\alpha}\right)\Gamma\left(1-\frac{2}{\alpha}\right)$.

The probability density function (PDF) is then
\begin{equation}
\label{eq:pdfSINR}
f_Z( z | \alpha, \lambda, \zeta, \xi  )\! =\! \left( \frac{2 \lambda \zeta}{\alpha} z^{\frac{2}{\alpha}-1} \!+\! \xi \right)\! \exp\left( \!-\! \zeta \lambda z^{\frac{2}{\alpha}} \!-\! \xi z  \right).
\end{equation}

\subsection{Short Blocklength Messages}
Following \cite{ART:Durisi-PIEEE2016, Yang2014c}, we define the encoding/decoding procedures as follows.  First, the encoder maps $k$ information bits $\mathbf{B} = \{B_1,..., B_k\}$ into a codeword with $n$ symbols $\mathbf{S} = \{S_1,...,S_n\}$, satisfying the power constraint $\tfrac{1}{n}\sum_{m=1}^n |S_m|^2 \leq W_s$. Then, $\mathbf{S}$ is transmitted through the wireless channel generating $\mathbf{T} = \{T_1,...,T_n\}$ as the output. Finally, the decoder makes an estimate about the information bits based on $\mathbf{T}$, namely $\hat{\mathbf{B}}$, satisfying a maximum error probability constraint $\epsilon$.
Then, we denote $R^*(n,\epsilon)$ in bits per channel use (bpcu) as the maximal coding rate at finite blocklength (FB) which renders the largest rate $\tfrac{k}{n}$ where $k$ is the number of information bits and $n$ denotes the blocklength, whose error probability does not exceed $\epsilon$ \cite{Durisi2016}. Then, under quasi-static conditions $R^*(n,\epsilon)$ can be tightly approximated by \cite{Yang2014c}
\begin{equation}
R^{*}(n, \epsilon)\approx\! \sup\left\lbrace R\! :\! \Pr\left[ \log_2\left(1+ Z \right) < R \right]\! <\epsilon \right\rbrace. 
\end{equation}
For codes of $R = \tfrac{k}{n}$ bpcu, the outage probability in quasi-static fading is approximated as \cite{Yang2014c}
\begin{equation}\label{eq:outage}
\epsilon = 
\expectation{\operatorname{Q}\left( \sqrt{n}\, \frac{\log_2\left(1+ Z\right) - R}{\sqrt{V(Z)}} \right)}{Z},
\end{equation}
where $V(Z) = \left(1-(1+Z)^{-2}\right)(\log_2 e)^2$ is the channel dispersion and measures the stochastic variability of the channel relative to a deterministic one with the same capacity \cite{Polyanskiy2010a}. The above outage function can also be expressed as $\Pr [\mathrm{SINR} < \gamma_{th}]$, where $\gamma_{th}$ is the SI(N)R threshold of the receiver which is determined by the channel capacity and is the minimum SI(N)R which is needed in order to have a successful link connection, then the reliability can be defined as $1-\Pr [\mathrm{SINR} < \gamma_{th}]$. In other words, an outage event occurs when a transmitted message is not successfully decoded by the receiver.   

\section{Outage Analysis}\label{sect:tput_analysis}
%
%
In this section we focus on the outage probability of the network described in Section~\ref{sect:smodel}. 
The analysis is done for two different scenarios, which are special cases of the general outage expression to be presented first.
The outage probability in \eqref{eq:outage} is intricate to be evaluated in closed-form, especially when considering a general SINR distribution as in \eqref{eq:pdfSINR}.
Therefore, we resort to a tight approximation of \eqref{eq:outage} before evaluating it in closed-form for the scenarios under investigation in this work. 

\begin{proposition}
	\label{prop:gen_outage}
	Given the network described in Section~\ref{sect:smodel}, the outage probability of the reference link (the link between the reference receiver and its respective transmitter \cite{elsawy2013stochastic}) is well approximated as
	%
	\begin{equation}\label{eq:general_outage}
	\epsilon_{ap} \!=\! \frac{F_Z(\vartheta)+F_Z(\varrho)}{2} \!+\! \frac{\beta \theta \left( F_Z(\vartheta)-F_Z(\varrho) \right)}{\sqrt{2 \pi}} 
	-\!\int\limits_{\varrho}^{\vartheta} \frac{\beta z}{\sqrt{2 \pi}} f_Z(z)\mathtt{d}z ,
	\end{equation}
	where $\beta = \sqrt{\frac{n}{2\pi}} (2^{2 R}-1)^{-\frac{1}{2}}$, $\theta = 2^R -1$, $\vartheta = \theta + \sqrt{\frac{\pi}{2}\beta^{-2}} $ and $\varrho = \theta - \sqrt{\frac{\pi}{2}\beta^{-2}}$.
\end{proposition}
\begin{proof}
	See Appendix~\ref{app:gen_outage}.
\end{proof}

Note that \eqref{eq:general_outage} covers a wide range of scenarios with path-loss exponent $\alpha > 2$. It is worth mentioning that only one integral remains in \eqref{eq:general_outage}, and the overall expression is composed of well known functions, which facilitates its integration by numerical methods compared with the original expression in \eqref{eq:outage}. 

\subsection{Dynamic Spectrum Access Interference-Limited Scenario}

\begin{figure}[tb]
	\centering
	\includegraphics[width=1\columnwidth]{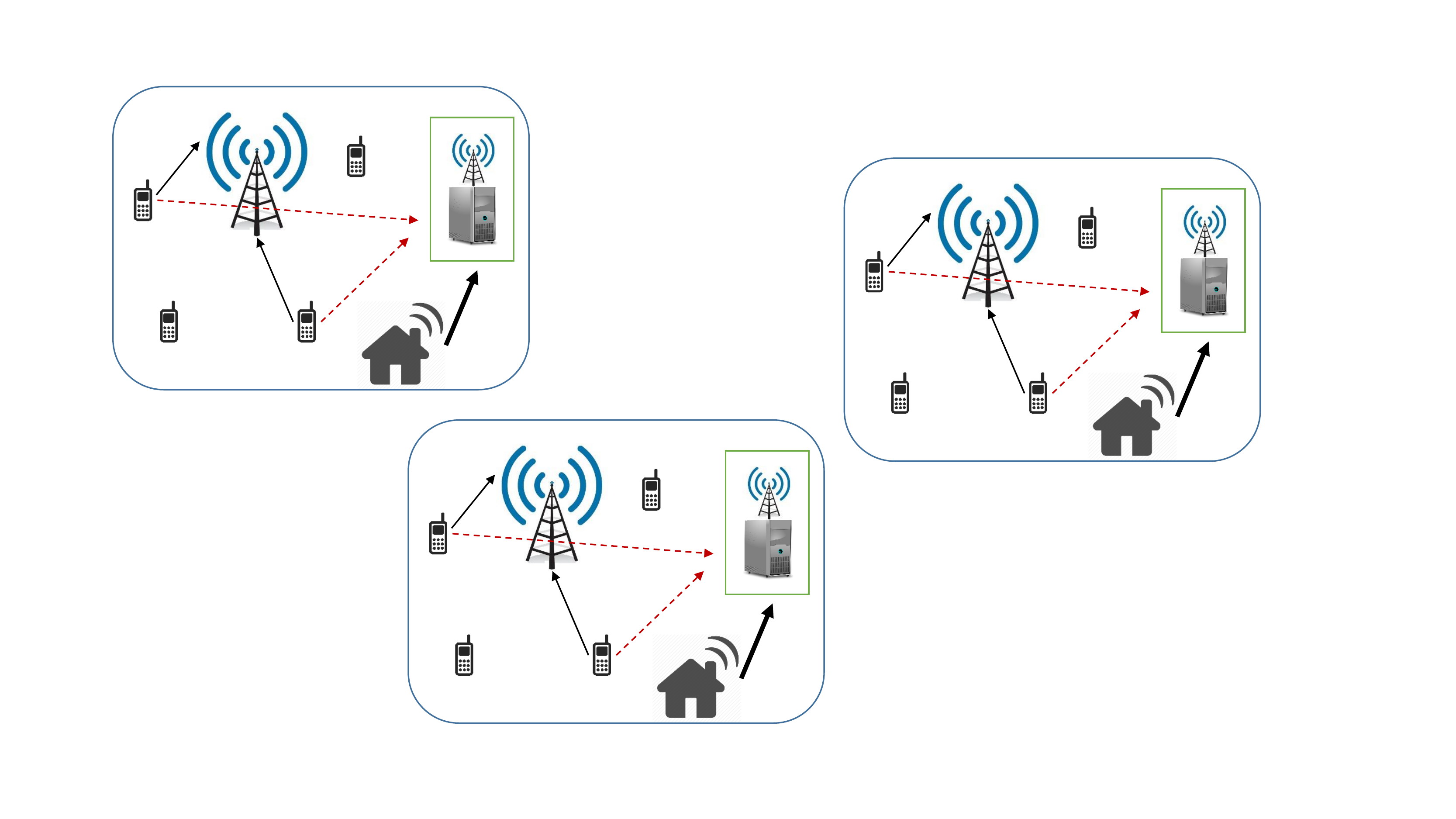}
     \caption{An illustration of the dynamic spectrum access scenario, where licensed and unlicensed users share the up-link channel. The reference smart meter (unlicensed transmitter) is depicted by the house, the aggregator (unlicensed receiver) by the CPU and its antenna, the handsets are the mobile licensed users (interferers to the aggregator) and the big antenna is the cellular base-station. As the smart meter uses directional antennas with limited transmit power (bold arrow), its interference towards the base-station can be ignored. The thin black arrows represent the licensed users' desired signal, while the red ones represent their interference towards the aggregator.}
	\label{fig:ss_mod}
\end{figure}

We consider a dynamic spectrum access model, shown in Fig. \ref{fig:ss_mod}, where licensed and unlicensed users share the frequency bands allocated to the uplink channel~\cite{tome2016joint}. 
We assume an interference-limited scenario where the licensed and unlicensed transmission powers are respectively $W_p$ and $W_s$. 
In this case, the noise power is negligible with respect to the aggregated interference~\cite{Nardelli2015}. 
The licensed users are mobile users communicating with a cellular base-station while the unlicensed users are the smart-meters that send data to their corresponding aggregator. 
The interference in this model is generated by mobile users with respect to the aggregator (reference receiver), as discussed in \cite{Nardelli2015}. 
%


\begin{proposition}
	\label{prop:gen_outagess}
     Assuming the network deployment described in Section \ref{sect:smodel} and the reference scenario in section \ref{sect:tput_analysis}-A, in an interference limited scenario where $\xi \approx 0$ , the outage probability is
	\begin{align}
	\label{eq:outage_SIR}
	\epsilon_\mathrm{SS}\! &=\!  
	\frac{\beta (\vartheta \!-\! \varrho)}{\sqrt{2\pi}} \!+\! \frac{\alpha \beta}{(\zeta\lambda)^{\frac{\alpha}{2}} 2 \sqrt{2\pi}} \nonumber\\  &\times
	\bigg[\Gamma\bigg(\frac{\alpha}{2}\,,\, \zeta \lambda \vartheta^{\frac{2}{\alpha}} \bigg) \!-\! \Gamma\bigg(\frac{\alpha}{2}\,,\, \zeta \lambda \varrho^{\frac{2}{\alpha}}\bigg)\bigg].
	\end{align}
\end{proposition}
\begin{proof}

	See Appendix~\ref{app:gen_outagess}.	
\end{proof}
\begin{corollary}
	\label{ssou}
	For the special case of interference-limited dense urban scenarios, where $\alpha=4$ is a good approximation for the path loss exponent,
	the outage probability in \eqref{eq:outage_SIR} reduces to
	\begin{align}\label{eq:outage_SIR4}
	&\epsilon_{\mathrm{SS}_{\alpha=4}}
	 =
	\frac{\beta (\vartheta-\varrho)}{\sqrt{2\pi}}+\frac{2\beta}{\sqrt{2\pi}(\zeta\lambda)^2}\nonumber\\
	&\times\bigg [  e^{-\zeta\lambda \sqrt{\vartheta}} \left(\frac{\zeta\lambda\vartheta}{\sqrt{\vartheta}}+1 \right)  - e^{-\zeta\lambda \sqrt{\varrho}} \left(\frac{\zeta\lambda\varrho}{\sqrt{\varrho}}+1 \right)\bigg ].
	\end{align}
\end{corollary}

\begin{figure*}[tb]
	\centering
	\includegraphics[width=2\columnwidth]{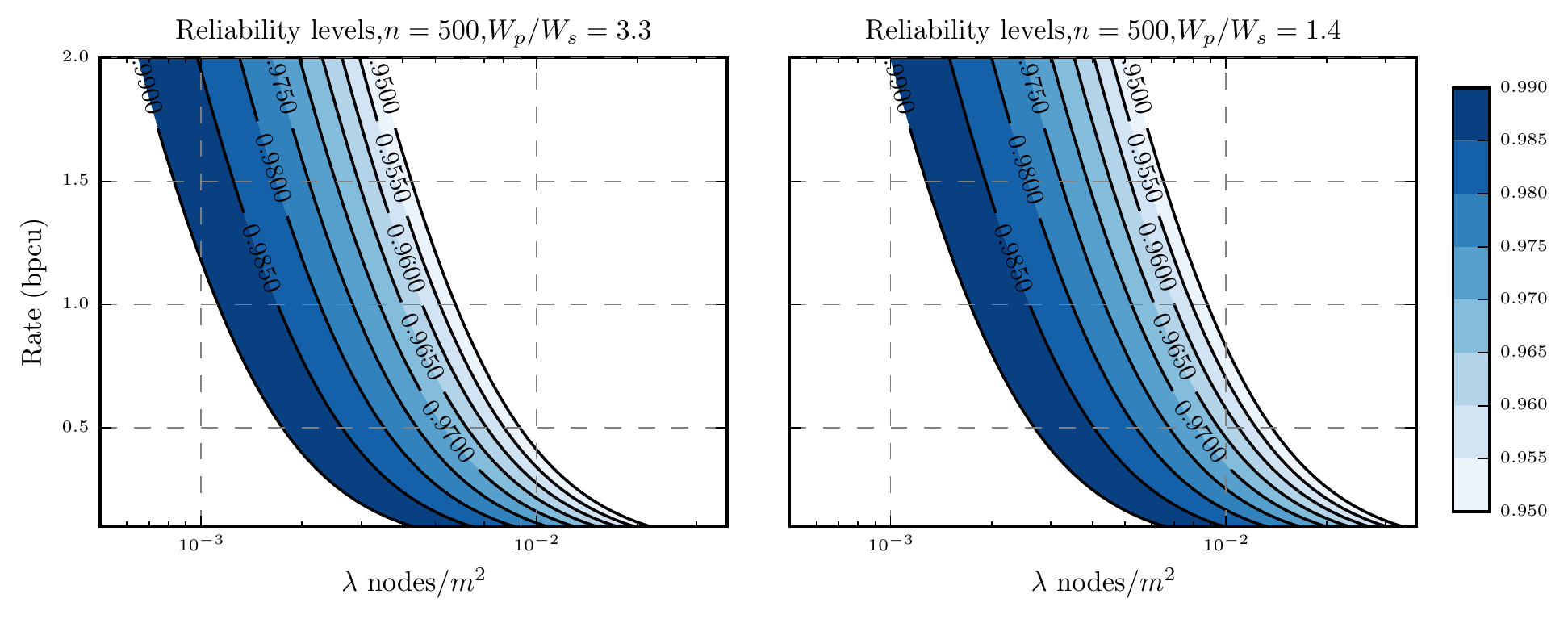}
	\caption{Dynamic spectrum access scenario outage probability $\epsilon$ as a function of the network density $\lambda$ and coding rate $R$, where $d=1$ and $\eta=1$. }
	\label{fig:outagectr}
\end{figure*}

The effect of different transmit powers $W_s$ on the outage probability as a function of the network density is shown in Fig. \ref{fig:outagectr}.
We can see that, as the transmit power increases, the link can reach a higher reliability level in denser networks even with a high rate.
It should be noted that this model is suitable for applications that do not require so strict reliability levels, such as smart meter reading.

As the noise level is negligible compared with the interference, the density of interfering nodes shall be high.
The mobile users -- the source of interference in this case -- transmit with a higher power compared with the smart meters.
%

\subsection{Locally Licensed Scenario}

\begin{figure}[tb]
	\centering
	\includegraphics[width=\columnwidth]{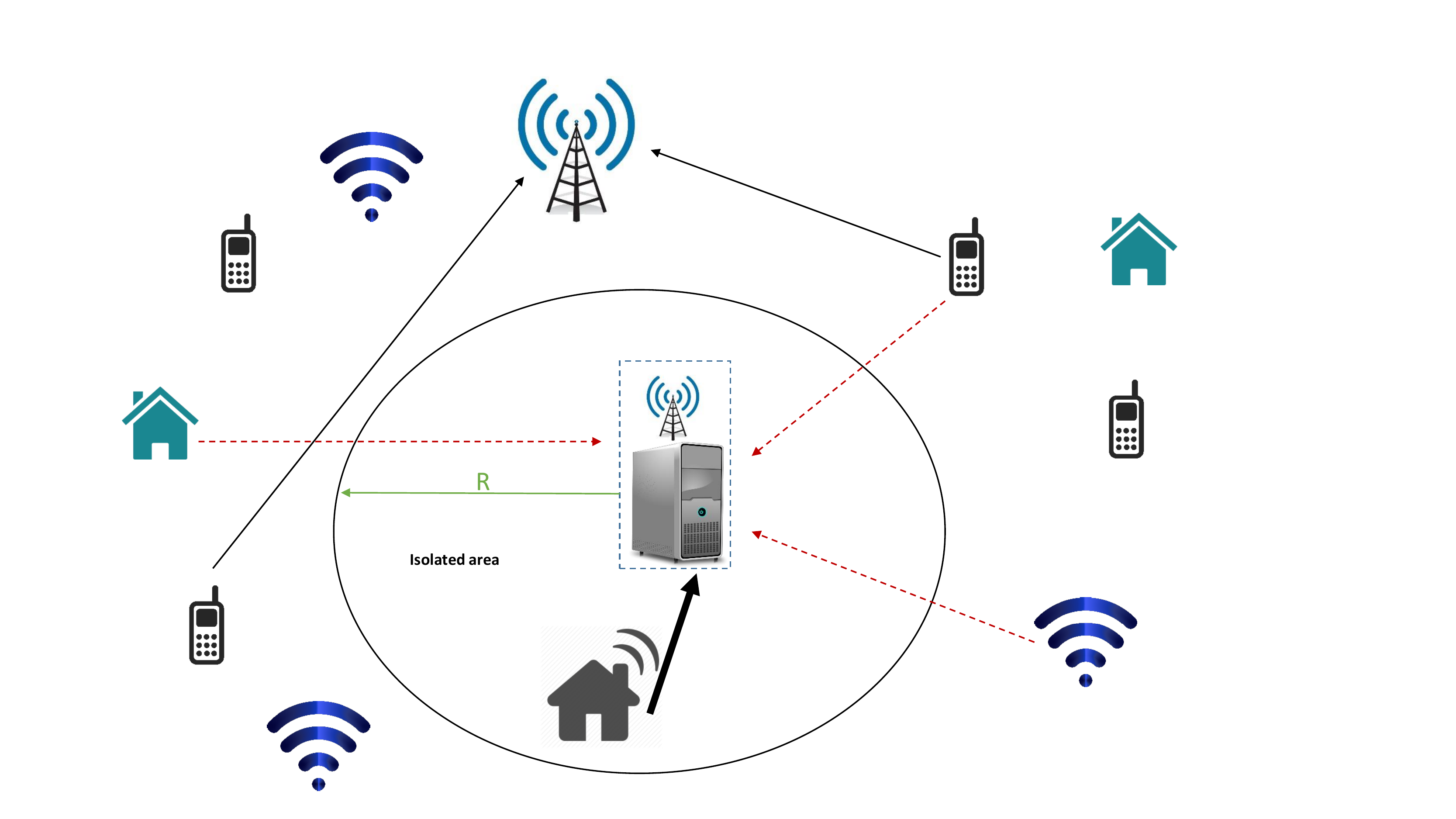}
	\caption{An illustration of the locally licensed scenario, where the micro operator holds an exclusive license for its own usage in a specific geographical region (isolated area). There is interference caused by the entities outside of this area such as mobile phones and Wi-Fi. However, since the level of interference in this area is very low, noise is what is going to harm the communications in this model. The reference smart meter is depicted by the house, the micro-operator  by the CPU and its antenna, the handsets are the mobile licensed users (interferers to the aggregator) and the big antenna is the cellular base-station. The thin black arrows represent the users' desired signal, while the red ones represent the interference coming from outside of the area.}
	\label{fig:mo_mod}
\end{figure}
In this case, we analyze a locally licensed scenario, using the $\mu\mathrm{O}$ concept, where the previous unlicensed link is now also a licensed user in the system in a specific geographical region.
Although this concept guarantees the exclusive usage of the frequency band, 
the  environment is still unfriendly: there are different entities that may cause interference such as base stations, smart meters and mobile users.
Therefore, instead of assuming no interference, we consider a point process with low density $\lambda$, leading to low interference levels. However, unlike the previous case, the noise level is not negligible in this case anymore and it will affect the reliability of the system.
The outage probability is given next.

\begin{figure}[tb]
	\centering
	\includegraphics[width=\columnwidth]{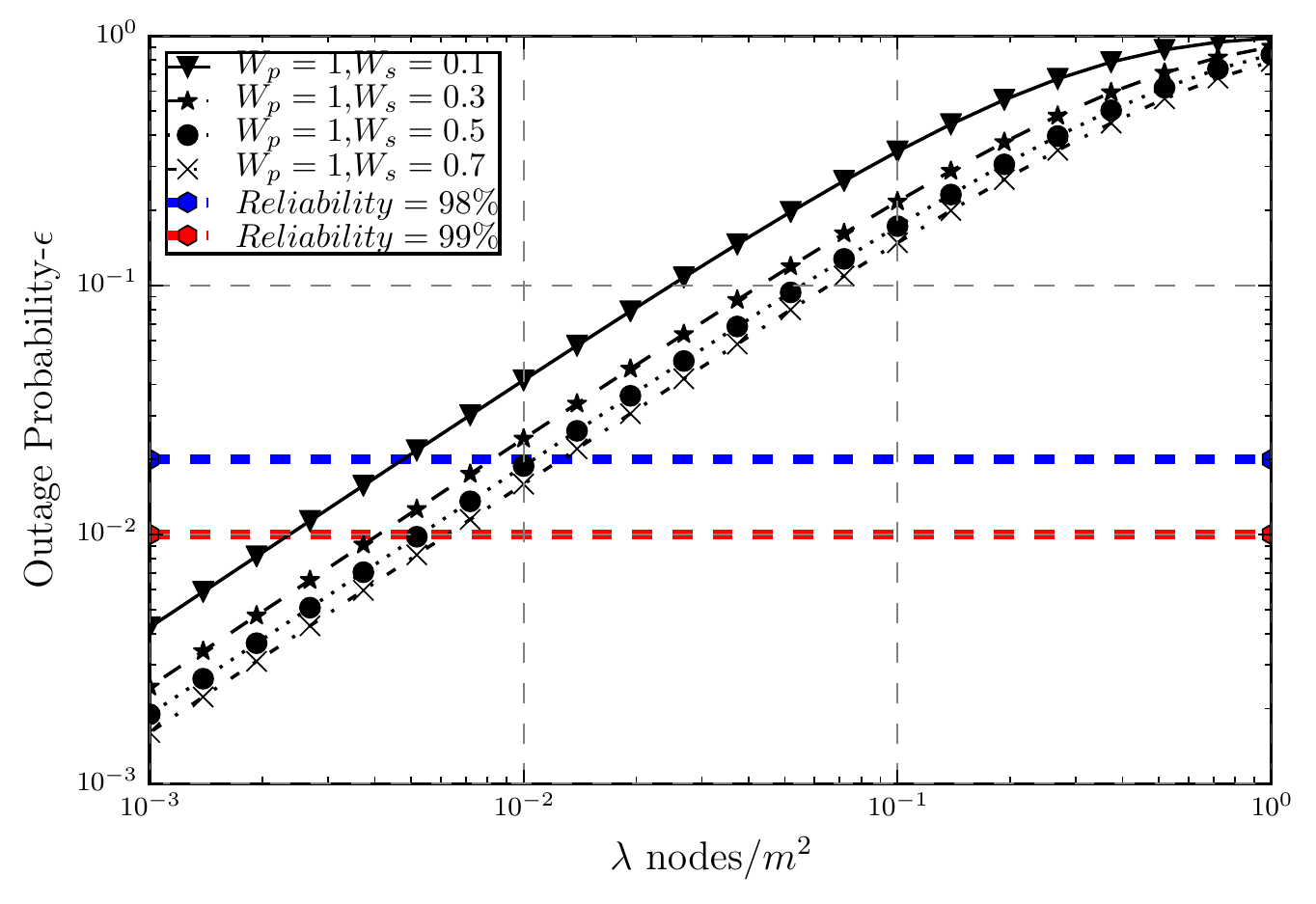}
	\caption{Dynamic spectrum access scenario outage probability  $\epsilon$ with different transmit power as a function of the network density $\lambda$, considering $d=1$, $\eta=1$, $n=200$,  and $R=0.1$. }
	\label{fig:outage_w}
\end{figure}

\begin{figure}[tb]
	\centering
	\includegraphics[width=\columnwidth]{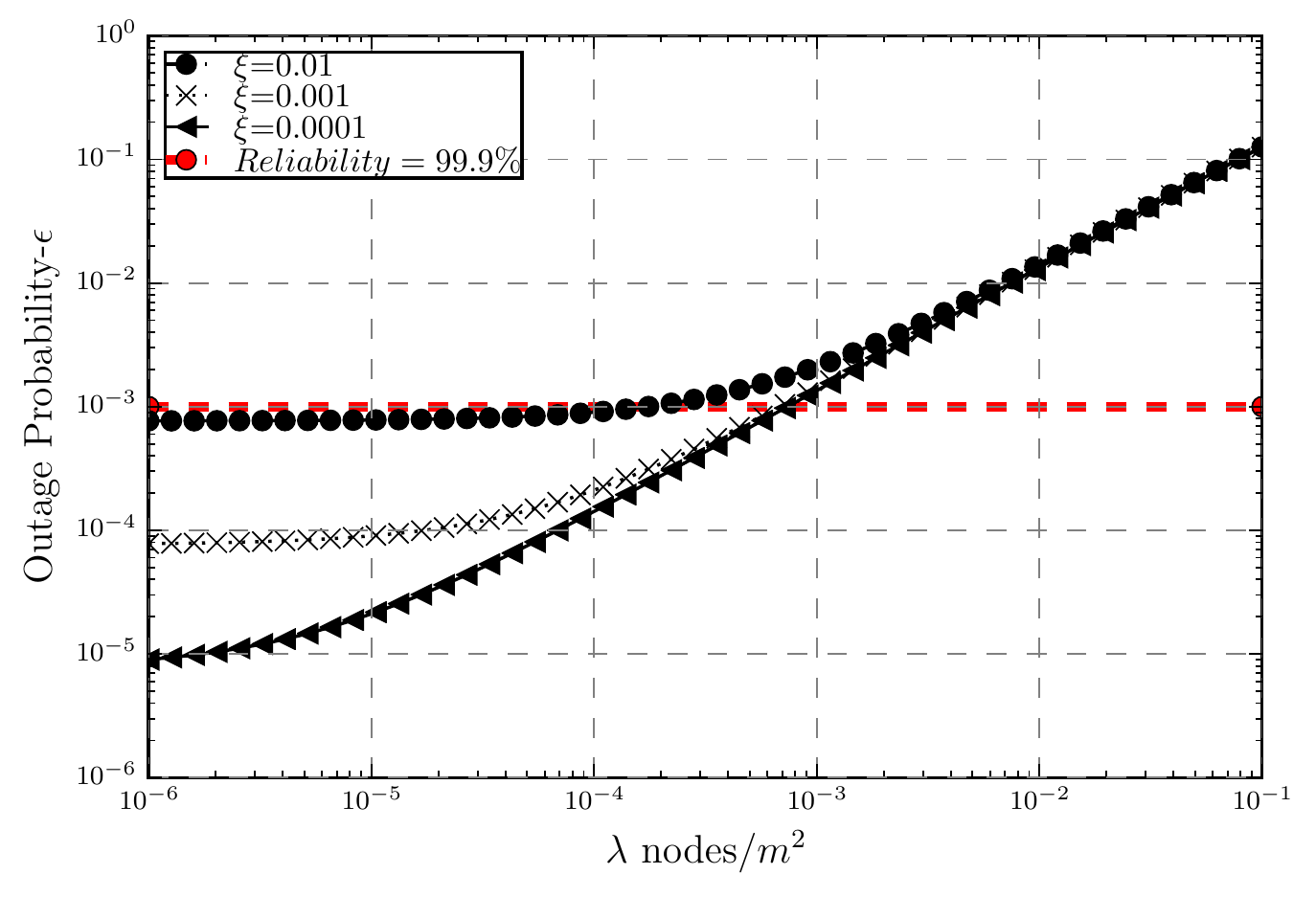}
	\caption{Locally licensed scenario outage probability  $\epsilon$ with different noise levels as a function of the network density $\lambda$, considering $d=1$, $n=200$, and $R=0.1$. }
	\label{fig:sc2ou}
\end{figure}

\begin{proposition} \label{prop:outage4}
	For the network model described in Section \ref{sect:smodel} and the characteristics of the $\mu\mathrm{O}$ scenario described in Section \ref{sect:tput_analysis}-B, the outage probability is given in \eqref{eq:outage_4mo} on top of page 8 where $\alpha=4$, hence, \eqref{eq:cdfSINR} is denoted by $F_Z(\vartheta |\xi)$ .
	\begin{figure*}[!t]
		\begin{align}\label{eq:outage_4mo}
		\overline{\epsilon}_{u
			o} &= 
		[1-F_Z(\vartheta |\xi)](\frac{-1}{2}-\frac{\beta \theta }{\sqrt{2\pi}})+[1-F_Z(\varrho |\xi)](\frac{-1}{2}+\frac{\beta \theta }{\sqrt{2\pi}})\nonumber
		-\frac{[1-F_Z(\varrho |\xi)][1-F_Z(\vartheta |\xi)]}{2\sqrt{2\pi}\xi^\frac{3}{2}}\nonumber\\
		&\times\bigg[2\beta\sqrt{\xi}((F_Z^{-1}(\vartheta |\xi)(1+\xi\varrho)-(F_Z^{-1}(\varrho |\xi)(1+\xi\vartheta) )\nonumber                                       
		+\exp(\zeta\lambda(\sqrt{\varrho}+\sqrt{\vartheta})+\xi(\varrho+\vartheta)+\frac{\zeta^2\lambda^2}{4\xi}   ) \sqrt{\pi}\beta\zeta\lambda       \nonumber \\
		&\times\left[\operatorname{erf}\left( \frac{\zeta \lambda  }{2\sqrt{\xi}} +\sqrt{\varrho\xi}\right) - \operatorname{erf}\left( \frac{\zeta \lambda  }{2\sqrt{\xi}}+{\sqrt{\vartheta\xi}} \right)   \right] \bigg].
		\end{align}
		\hrule 
	\end{figure*}
\end{proposition}

\begin{proof}
	See Appendix~\ref{app:gen_outage4}.
\end{proof}

The $\mu\mathrm{O}$ scenario cannot be assumed interference-limited; on the contrary, the noise power here is the major factor in the SINR.
%
%
To compute the numerical results, we assume here $W_p = W_s=1$ and $\xi=0.001$. 
Recall that in the interference limited scenario, when the interference is low, we can achieve a low outage probability even with a high network density for a given coding rate. 
Increasing the smart meters transmit power results in having a higher success probability. 
As the interference power increases, a lower coding rate is needed  to have the same outage probability with the same density. 

In terms of outage, the $\mu\mathrm{O}$ scenario, in its turn, behaves similarly, meaning that with increasing the network density, the outage probability increases, however, the outage probability is generally lower in this model as the interference level is negligible and noise is the main factor affecting the performance of the network.
The operating regions for the dynamic spectrum access and the $\mu\mathrm{O}$ scenarios are shown in Figs. \ref{fig:outage_w} and \ref{fig:sc2ou}, respectively.
The presented outage levels are chosen based on the fact that the dynamic spectrum access model is suitable for the non-critical applications, which we illustrate by reliabilities between $98\%$ to $99\%$, which is relatively high for some applications (for smart grid application, refer to \cite[Table.3]{kuzlu2014communication}, and will be further discussed in the next section).

On the other hand, when we assess the $\mu O$, one can see a very high reliability with error probability as low as $0.1\%$ (i.e. reliability $ \geq 99.9\%$ ).
In other words, $\mu O$ is suitable for critical applications with high reliability requirements. 
Thus, the operating region for this model is presented as the area where $\overline{\epsilon}_{u 	o} \leq 10^{-3}$.
 
\subsection {On the Accuracy of ($6$) - ($9$)} \label{khar}

As it was mentioned earlier in this section, we use an approximation of \eqref{eq:outage} for calculating the outage probability in \eqref{eq:general_outage} which is then used to derive the closed form equations of the outage probability for different scenarios presented in \eqref{eq:outage_SIR}-\eqref{eq:outage_4mo}. In this section, the approximation is compared to the exact equation for both the interference and noise limited scenarios, as shown in Figs. \ref{f5} and \ref{fig2} respectively. We can see that the results from the approximated and closed form equations are almost always equal to the exact equation \eqref{eq:outage}. Considering the error metric below 

 \begin{equation}
 \label{eq:error}
 \Delta=|\frac{\epsilon-\epsilon_{ap} }{\epsilon} |,
 \end{equation}
 
\noindent we can see that $\Delta$ is almost always either zero or very close to zero. It is only for the case of the locally licensed scenario that we can see, for a very low values of $\lambda$, a difference of at most $4\%$ between \eqref{eq:outage} and the approximation which is still a very low difference. A more elaborated analysis of the error metric can be found in \cite{lopez2017wireless}.

\begin{figure}[]
	\centering
	\includegraphics[width=\columnwidth]{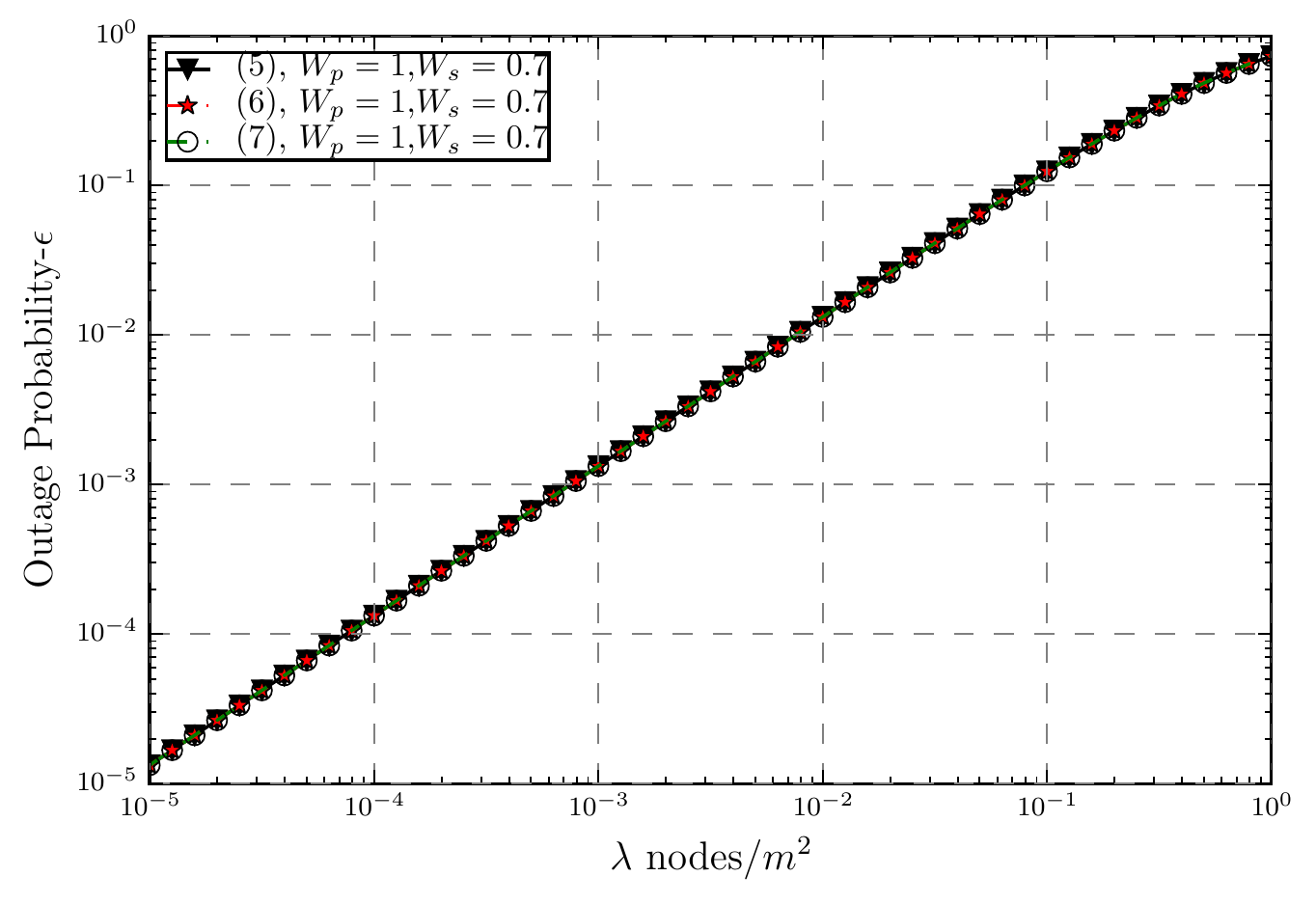}
	\caption{Comparison of the accuracy of the approximation used in (6) and (7) compared to (5) for the Dynamic spectrum access scenario outage probability  $\epsilon$ as a function of the density $\lambda$, considering $d=1$, $\eta=1$, $n=500$, and $R=0.1$. }\label{f5}
	
\end{figure}

\begin{figure}[tb]
	\centering
	\includegraphics[width=\columnwidth]{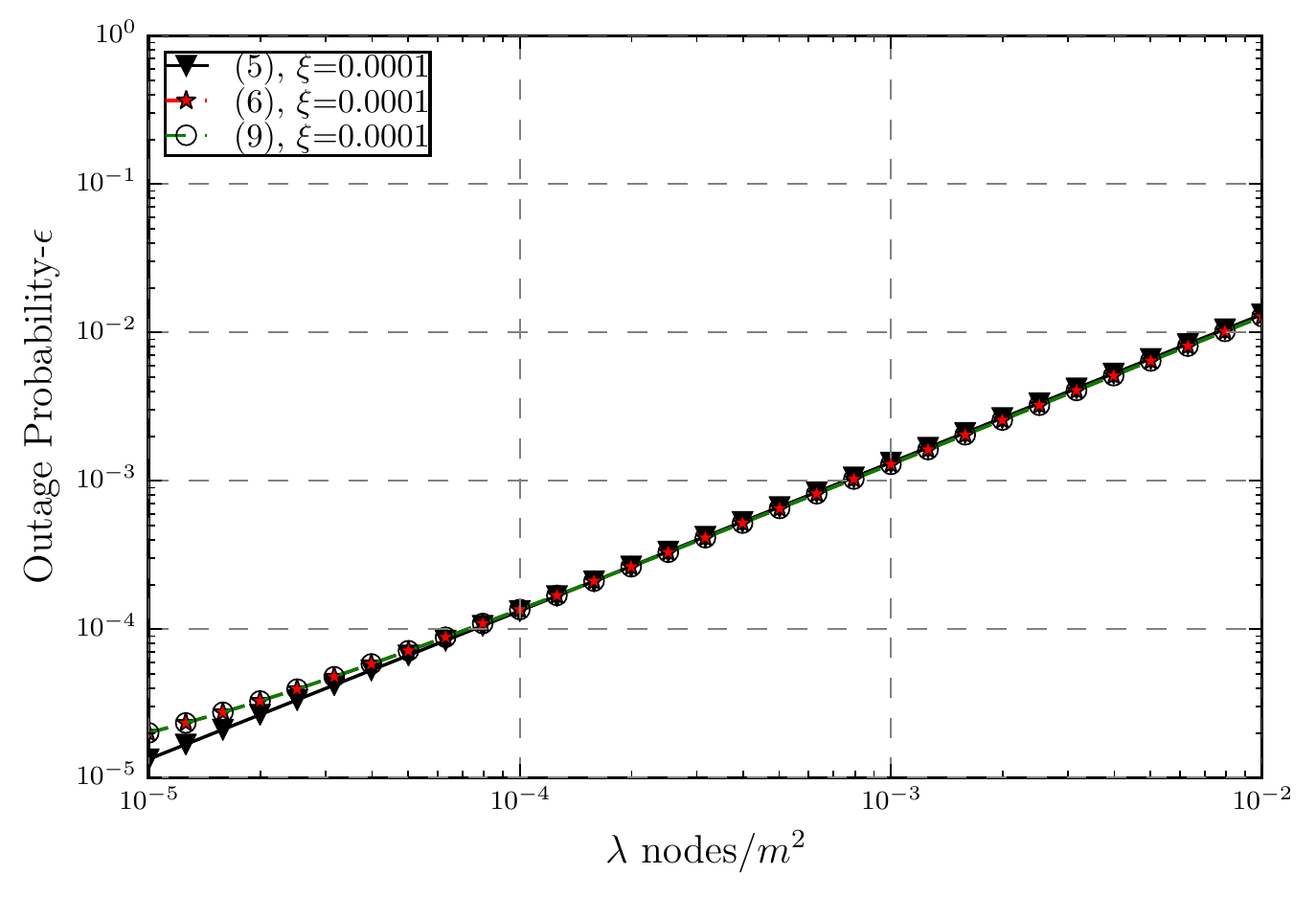}
	\caption{Comparison of the accuracy of the approximation used in (6) and (9) compared to (5) for the locally licensed scenario outage probability  $\epsilon$ as a function of the density $\lambda$, considering $d=1$, $n=500$, and $R=0.1$. }\label{fig2}
	
\end{figure}
 
\section {Meeting the Smart Grid Requirements}\label{inv}
This section focuses on the specific requirements for different smart grid applications.
Specifically, we analyze the impact of blocklength, retransmission attempts, and network density on the outage probability of the two proposed scenarios.
%
	%
%




\subsubsection{Dynamic Spectrum Access Scenario}

This scenario is suitable for applications like smart meters periodic transmissions  \cite{kuzlu2014communication}.
The frequency of the transmissions might be relatively high over one day \cite{tome2016joint}: the smart meters transmit data every $15$ minutes during a period of $24$ hours which means smart meters need to transmit data $96$ times per day.

The properties explained above can be seen in Fig.\ref{fig:t_w}-a where the behavior of $\lambda (\epsilon_{SS} |\alpha=4)$ is shown. 
With increasing $\lambda$, the outage probability also increases due to the higher interference level.
Nevertheless, we can still achieve our desired outage probability, even in denser networks. 
We can confirm that as the unlicensed users transmit power increases, we can achieve better reliability in denser networks.

It was mentioned earlier that with finite blocklength, we cannot achieve the ultra-reliable (UR) region with the dynamic spectrum access model.
Consequently, this approach shall be used in applications with looser requirements ($98\%-99\%$).
In Fig.\ref{fig:nsub}-a we investigate the effect of the finite blocklength ($n$) on the outage probability for different information bit sizes ($k$).
where we can see that, for $100 \leq n \leq 1000$,  $\epsilon \leq 10^{-3} $ cannot be achieved.
For reaching the UR region, the blocklength would have to be increased to very large numbers.
\begin{figure*}[tb]
	\centering
	\includegraphics[width=2\columnwidth]{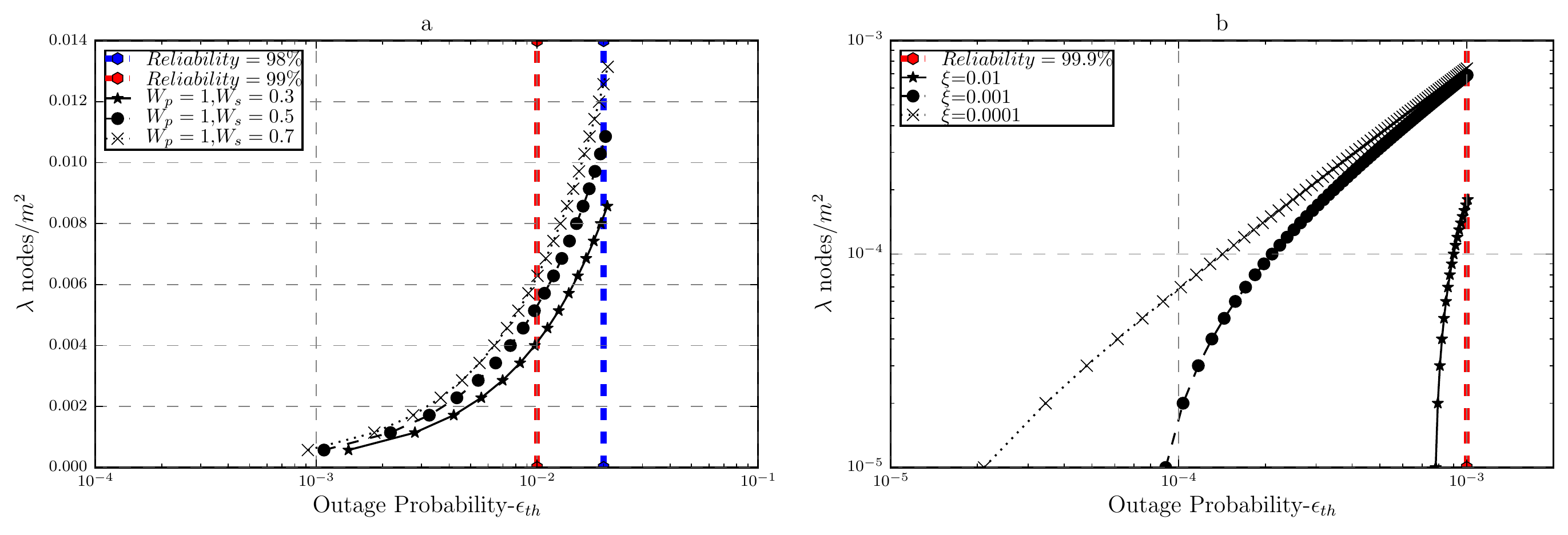}
	\caption{Network density $\lambda$ as a function of the outage probability $\epsilon$ in (a) dynamic spectrum access scenario where $d=1$, $\eta=1$, $n=200$,  and $R=0.1$ and (b) locally licensed scenario, considering $d=1$, $n=200$, $R=0.1$ and $\frac{W_p}{W_s}=1$. }
	\label{fig:t_w}
	\vspace{2ex}
\end{figure*}
%
%
\subsubsection{Locally Licensed Scenario}

$\mu\mathrm{O}$  can achieve higher levels of reliability, meeting the requirements of more critical applications like fault detection. 
For more examples about the reliability and data size requirements for different smart grid applications, refer to \cite[Table.3]{kuzlu2014communication}. 
The $\mu\mathrm{O}$ approach is subjected to a lower interferers' density when compared to the dynamic spectrum access, as shown in Fig.\ref{fig:t_w}-b. 
When the noise level is low, the outage probability is also low. 
It is illustrated in this figure that for $\lambda < 10^{-3}$, this model can operate in the UR region where reliability level is at least $99.9\%$.

%
%
%
\begin{figure*}[tb]
	\centering
	\includegraphics[width=2\columnwidth]{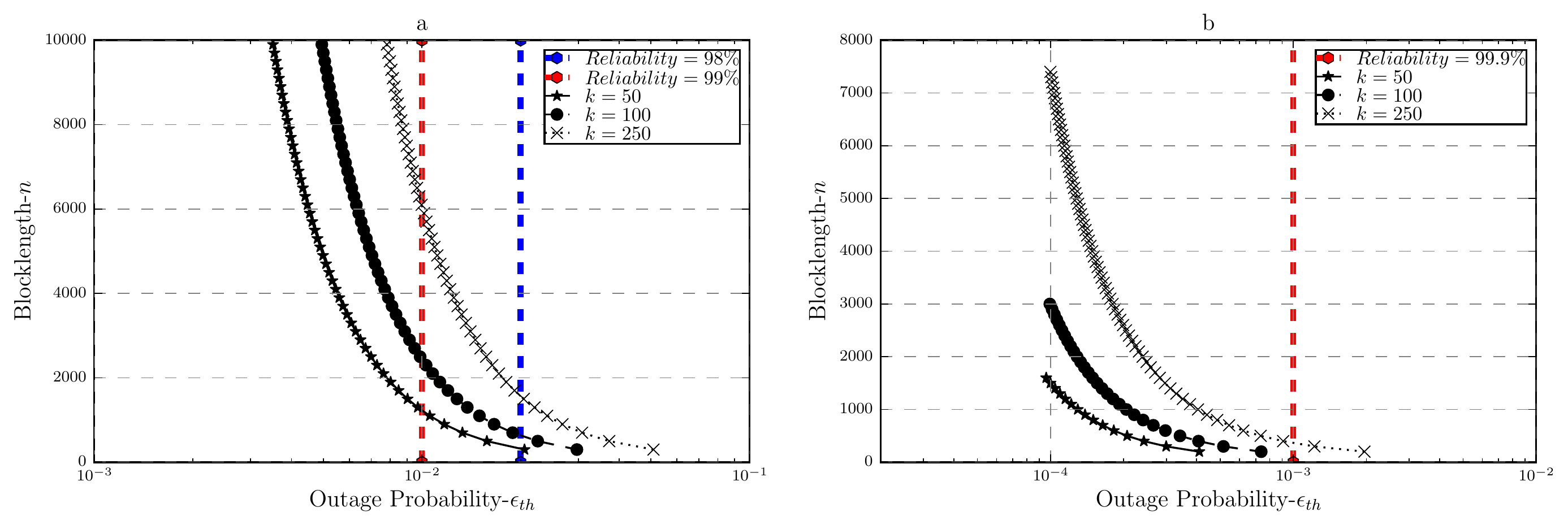}
	\caption{Blocklength $n$ as a function of the outage probability $\epsilon$ in (a) dynamic spectrum access scenario, considering $d=1$, $\eta=1$, $n=200$, $R=0.1$, $\lambda=10^{-2}$ and $\frac{W_p}{W_s}=1.4$ and (b) locally licensed scenario, considering $d=1$, $n=200$, $R=0.1$, $\xi=0.001$ Watts, $\lambda=10^{-4}$ and $\frac{W_p}{W_s}=1$}
	\label{fig:nsub}
	\vspace{2ex}
\end{figure*}

\begin{figure*}[!t]
	\centering
	\includegraphics[width=2\columnwidth]{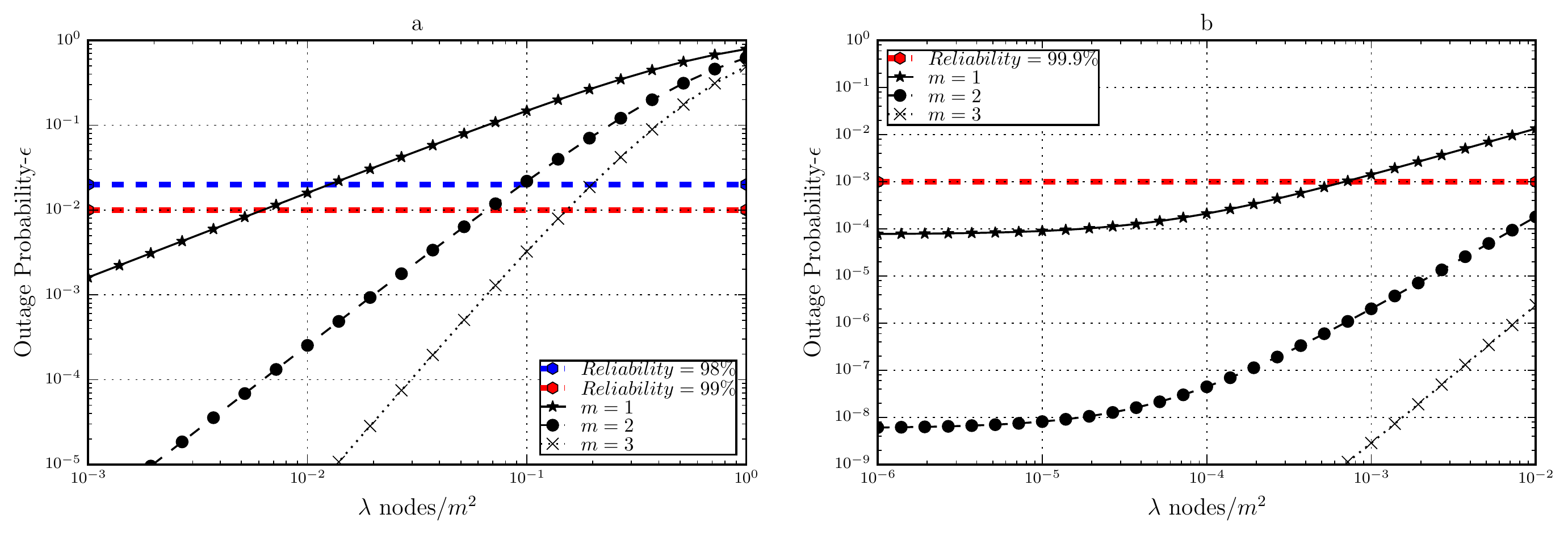}
	\caption{Effect of increasing the number of retransmissions on (a) dynamic spectrum access scenario and (b) locally licensed scenario outage probability $\epsilon$ as a function of the network density $\lambda$, considering $d=1$, $\frac{W_p}{W_s}=1.4$ and $\xi=0.001$.}
	\label{fig:ret}
\end{figure*}

%
Fig.\ref{fig:nsub}-b shows that ultra-reliability can be achieved with short blocklength in the $\mu\mathrm{O}$ scenario.
By increasing $k$, the required blocklength for keeping the link in the UR region also increases; in any case, $\epsilon \leq 10^{-3}$ is still achievable for relatively high $k$ when  $100 \leq n \leq 1000$.

\vspace{-2ex}
\subsection{Retransmission Attempts}

Two basic strategies are normally employed to cope with transmission errors in communication systems, namely automatic repeat request (ARQ) and forward error correction (FEC)\cite{lin1982hybrid}.
In this paper, we only consider ARQ.

ARQ consists of an acknowledgment (ACK) or negative acknowledgment (NACK) messages to be sent by the receiver to inform whether the intended message has been successfully decoded.
If the transmission was not successful, a retransmission is requested and the retransmission continues until the codeword is decoded successfully or the allowed maximum number of retransmission is reached \cite{Makki2014a}. 
This strategy has some drawbacks such as loss of throughput, which are studied in \cite{arq1,arq2}  (without considering the high reliability or low latency, though).
Following \cite{dosti2017ultra}, we study the effect of the number of transmission attempts for a given message on the outage probability $\bar{\epsilon}$, including at most $m$ transmission attempts assuming Type-I HARQ is \cite{el2004performance}:
\begin{equation}\label{eq:outageh}
\bar{\epsilon} = \epsilon (n,\lambda) ^m,
%
\end{equation}
\noindent where $\epsilon$ is the outage probability in \eqref{eq:outage}.

The effect of increasing the number of transmission attempts on both scenarios is shown in Fig. \ref{fig:ret}. It can be seen that as $m$ increases, the reliability is enhanced. 
Comparing the outage probabilities of when only one or up to two transmission attempts are allowed, we can see that for the same $\lambda$, a much lower outage probability can be achieved. 

Considering the outage curve for dynamic spectrum access and $\mu\mathrm{O}$, at some point at very low interferers' densities, the two curves cross each other. This is due to the fact that the network density becomes very low from that point onward.
Therefore, the interference power becomes lower and the dynamic spectrum access model starts to have a lower outage than the $\mu\mathrm{O}$ scenario.
However, it is important to remember that the dynamic spectrum access model is designed to be used in denser networks; so the fact that its outage probability becomes less than the $\mu\mathrm{O}$ scenario for networks with very low densities shall be neglected since it contradicts the basic assumption of a interference-limited network.

While retransmissions increase the reliability of the network, it also increases the latency, which is in fact another important aspect of MTC.
As 5G and MTC have strict requirements in terms of latency, we should also consider this metric in our analysis by limiting the maximum number of transmission attempts.. 
In this case, the total delay is calculated as
\begin{equation}\label{eq:late}
\delta=n+\nu+\sum_{j=1}^{m} (n_j+\nu_j),
\end{equation}
where $\nu$ denotes the number of channel uses for ACK/NACK messages that have been sent and $n$ is the number of channel uses.
The average delay expression is then
\begin{equation}\label{eq:lateav}
\overline{\delta}=\frac{1}{m}\sum_{j=1}^{m}\epsilon (n,\lambda).
\end{equation}

The number of channel uses and symbol time are the determining factors when dealing with delay. 
The symbol time ($T_s$) of LTE (long term evolution) is $T_s\approx 66.7 \mu s$ \cite{seidel2007overview} and the current latency requirements of different smart grid applications are described in \cite[Table.3]{kuzlu2014communication}. However, 5G is going to benefit from ultra low latency compared to LTE. Hence, the smart grids are also going to have ultra low latency which is expected to be $3 ms$ to $5 ms$ \cite{ts}. Considering that the symbol time is going to be $T_s= \frac{1}{120k}\approx 8.3 \mu s$, for $n=200$ and $m=1$, $\delta=1.66 ms + \nu$. As the number of allowed transmission attempts increases, the delay also increases. If $m=2$, then $\delta=3.32 ms + \nu$. We can see that increasing $m$ results in increasing the delay and this increase will not be linear since at each transmission there is also a feedback message sent every time. Thus, it is very important to limit the number of transmission attempts to avoid increasing the latency of the system.

%
%

%
%
%
%

\section{Discussion and final remarks}\label{sect:conclusion}

This paper evaluates the possibility of meeting the reliability requirements of different smart grid applications by using FB in two different system models, (\textit{i}) dynamic spectrum access scenario suitable for applications with loose reliability requirements ($98\%-99\%$), and (\textit{ii}) $\mu\mathrm{O}$ scenario suitable for applications with strict reliability requirements (more than $99\%$).

Our results show that it is possible to meet the expected reliability levels and even reach the UR region for smart grids while having FB. It is shown that several factors such as network density, coding rate and interference and noise level affect the outage probability of the system, hence, they should be taken into consideration when choosing a suitable model considering the required reliability level of a specific application. Studying the ultra reliability and delay opens up a wide range of research opportunities in the smart grid communication systems.

\appendices
\section{Proof of Proposition~\ref{prop:gen_outage} } \label{app:gen_outage}
As pointed out in \cite{Makki2014a}, the function $\operatorname{Q}(g(t))$ can be tightly approximated by a linear function for the whole S(I)NR range. Notice that the argument inside the $\operatorname{Q}$-function in \eqref{eq:outage} is given as $g(t) = \sqrt{n}  \left(1-(1+ t)^{-2}\right)^{-\frac{1}{2}} \log\left(1+ t\right)$, and that $g(t)$ is an increasing function of $t$, but not strictly positive $\forall t \in \mathbb{R}$, which restricts the use of other well known approximations for the $\operatorname{Q}$-function \cite{Nguyen2015}.
Then, let $\operatorname{Q}(g(t)) \approx W(t)$ be denoted as
\begin{align}
\label{eq:z}
W(t) = 
\begin{cases} 
1 & t \leq \varrho  \\
\frac{1}{2} - \frac{\beta}{\sqrt{2 \pi}} (t-\theta)   & \varrho < t < \vartheta \\
0 & t \geq \vartheta 
\end{cases}
\end{align}
where $\theta = 2^R-1$ is the solution of $g(t)=0$, while $\beta=\sqrt{\frac{n}{2\pi}} (2^{2 R}-1)^{-\frac{1}{2}}$ is the solution for $\frac{\partial \operatorname{Q}(g(t))}{\partial t} |_{t=\theta}$. 

Then, the outage probability becomes, 
\begin{align}
\epsilon &= \mathbb{E}_{Z}\left[\epsilon\right] = \int_{0}^{\infty} W(z) f_Z(z) \mathtt{d}z\nonumber
	\nonumber \\
&= \int_{0}^{\varrho}(f_Z(\varrho))\mathtt{d}z  + \int_{\varrho}^{\vartheta} \left( \frac{1}{2} - \frac{\beta}{\sqrt{2 \pi}} (z-\theta) \right) f_Z(z)
\mathtt{d}z	\nonumber
	\nonumber \\
&=F_Z(Z)\biggl |_0^\varrho +\int_{\varrho}^{\vartheta} \left( \frac{1}{2} - \frac{\beta}{\sqrt{2 \pi}} (z-\theta) \right) f_Z(z)\mathtt{d}z\nonumber\\
&=F_Z(\varrho)+\frac{1}{2}F_Z(z)\biggl |_\vartheta^\varrho + \frac{\beta\theta}{\sqrt{2 \pi}}F_Z(z)\biggl |_\vartheta^\varrho - \int\limits_{\varrho}^{\vartheta} \frac{\beta z}{\sqrt{2 \pi}} f_Z(z)\mathtt{d}z,
\end{align}
which after few algebraic manipulations is written as in \eqref{eq:general_outage}.

\section{Proof of Proposition~\ref{prop:outage4} } \label{app:gen_outagess}
By setting $\xi=0$ and replacing it into \eqref{eq:cdfSINR} and \eqref{eq:pdfSINR}, while considering  $F_Z (z) = F_Z( z | \alpha=4,\xi=0)$ and $f_Z (z) = f_Z( z | \alpha=4,\xi=0)$, the integral in \eqref{eq:general_outage} assumes the form of the upper incomplete gamma function \cite[§. 6.5.3]{BOOK:ABRAMOWITZ-DOVER03}.

\begin{align}
\label{eq:integral_Ass}
I&= \frac{\beta}{\sqrt{2 \pi}} \int_{\varrho}^{\vartheta} z f_Z(z)\mathtt{d}z\nonumber \\
&=\frac{\beta}{\sqrt{2 \pi}} \int_{\varrho}^{\vartheta} z( \frac{2 \lambda \zeta}{\alpha} z^{\frac{2}{\alpha}-1})\! \exp( \!-\! \zeta \lambda z^{\frac{2}{\alpha}}   )\mathtt{d}z\nonumber\\
& =\frac{2\beta\lambda\zeta}{\sqrt{2 \pi}\alpha}  \int_{\varrho}^{\vartheta} z (z^{\frac{2}{\alpha}-1})\exp( \!-\! \zeta \lambda z^{\frac{2}{\alpha}}   )\mathtt{d}z\nonumber\\                                      
&=\frac{2\beta\lambda\zeta}{\sqrt{2 \pi}\alpha}  \int_{\varrho}^{\vartheta} z^{\frac{2}{\alpha}}\exp( \!-\! \zeta \lambda z^{\frac{2}{\alpha}}   )\mathtt{d}z,
\end{align}

\noindent where the integral can be solved using integration by parts \cite[§2.02-5]{BOOK:GR2007}. Hence, we attain \eqref{eq:outage_SIR} by replacing \eqref{eq:integral_Ass} into \eqref{eq:general_outage} after few algebraic manipulations and using \cite[§2.321-1]{BOOK:GR2007} and \cite[§2.33-10]{BOOK:GR2007}.

\section{Proof of Proposition~\ref{prop:outage4} } \label{app:gen_outage4}

Considering  $F_Z (z) = F_Z( z | \alpha=4,\xi)$ and $f_Z (z) = f_Z( z | \alpha=4,\xi)$, the integral in \eqref{eq:general_outage} has a closed-form solution when $\alpha=4$ as follows
\begin{align}
\label{eq:integral_A4}
I&= \frac{\beta}{\sqrt{2 \pi}} \int_{\varrho}^{\vartheta} z f_Z(z)\mathtt{d}z\nonumber \\
&=\frac{[1-F_Z(\varrho)][1-F_Z(\vartheta)]}{2\sqrt{2\pi}\xi^\frac{3}{2}}\nonumber\\
&\times\bigg[2\beta\sqrt{\xi}((1-F_Z(\vartheta))^{-1}(1+\xi\varrho)\nonumber \\
&-((1-F_Z(\varrho))^{-1}(1+\xi\vartheta)) )\nonumber\\                                       
&+\exp(\zeta\lambda(\sqrt{\varrho}+\sqrt{\vartheta})+\xi(\varrho+\vartheta)+\frac{\zeta^2\lambda^2}{4\xi}   ) \sqrt{\pi}\beta\zeta       \nonumber \\
&\quad\left[\operatorname{erf}\left( \frac{\zeta \lambda  }{2\sqrt{\xi}} +\sqrt{\varrho\xi}\right) - \operatorname{erf}\left( \frac{\zeta \lambda  }{2\sqrt{\xi}}+{\sqrt{\varrho\xi}} \right)   \right] \bigg],
\end{align}
where the integral can be solved though integration by parts and with the help of \cite[§2.33-10]{BOOK:GR2007} and \cite[§2.33-16]{BOOK:GR2007}. Then, substituting $\alpha=4$ and \eqref{eq:integral_A4} into \eqref{eq:general_outage}, yields \eqref{eq:outage_4mo}.

\bibliographystyle{IEEEtran}
\bibliography{IEEEabrv,refsURC}

\begin{thebibliography}{10}
\providecommand{\url}[1]{#1}
\csname url@samestyle\endcsname
\providecommand{\newblock}{\relax}
\providecommand{\bibinfo}[2]{#2}
\providecommand{\BIBentrySTDinterwordspacing}{\spaceskip=0pt\relax}
\providecommand{\BIBentryALTinterwordstretchfactor}{4}
\providecommand{\BIBentryALTinterwordspacing}{\spaceskip=\fontdimen2\font plus
\BIBentryALTinterwordstretchfactor\fontdimen3\font minus
  \fontdimen4\font\relax}
\providecommand{\BIBforeignlanguage}[2]{{%
\expandafter\ifx\csname l@#1\endcsname\relax
\typeout{** WARNING: IEEEtran.bst: No hyphenation pattern has been}%
\typeout{** loaded for the language `#1'. Using the pattern for}%
\typeout{** the default language instead.}%
\else
\language=\csname l@#1\endcsname
\fi
#2}}
\providecommand{\BIBdecl}{\relax}
\BIBdecl

\bibitem{ART:Ericsson2015}
\BIBentryALTinterwordspacing
Ericsson, ``{Ericsson mobility report on the pulse of the networked society},''
  \emph{Ericsson White Papers}, June 2015. [Online]. Available:
  \url{http://www.ericsson.com}
\BIBentrySTDinterwordspacing

\bibitem{ART:McKinsey2015}
J.~Manyika \emph{et~al.}, ``{Unlocking the Potential of the Internet of
  Things},'' \emph{McKinsey Global Institute}, 2015.

\bibitem{ART:Perera-A2014}
C.~Perera, C.~H. Liu, S.~Jayawardena, and M.~Chen, ``A survey on internet of
  things from industrial market perspective,'' \emph{IEEE Access}, vol.~2, pp.
  1660--1679, 2014.

\bibitem{ali2015next}
A.~Ali, W.~Hamouda, and M.~Uysal, ``Next generation {M2M} cellular networks:
  challenges and practical considerations,'' \emph{IEEE Communications
  Magazine}, vol.~53, no.~9, pp. 18--24, 2015.

\bibitem{tullberg2016metis}
H.~Tullberg, P.~Popovski, Z.~Li, M.~A. Uusitalo, A.~H{\"o}glund,
  {\"O}.~Bulakci, M.~Fallgren, and J.~F. Monserrat, ``The metis {5G} system
  concept--meeting the {5G} requirements,'' \emph{IEEE Communications
  Magazine}, vol.~54, no.~12, pp. 132--139, 2016.

\bibitem{bockelmann2016massive}
C.~Bockelmann, N.~Pratas, H.~Nikopour, K.~Au, T.~Svensson, C.~Stefanovic,
  P.~Popovski, and A.~Dekorsy, ``Massive machine-type communications in {5G}:
  Physical and mac-layer solutions,'' \emph{IEEE Communications Magazine},
  vol.~54, no.~9, pp. 59--65, 2016.

\bibitem{hu}
O.~L.~A. L{\'o}pez, H.~Alves, R.~D. Souza, and E.~M.~G. Fern{\'a}ndez,
  ``Ultrareliable short-packet communications with wireless energy transfer,''
  \emph{IEEE Signal Processing Letters}, vol.~24, no.~4, pp. 387--391, 2017.

\bibitem{onel2}
O.~Alcaraz~L{\'o}pez, E.~Fern{\'a}ndez, R.~Demo~Souza, and H.~Alves,
  ``Ultra-reliable cooperative short-packet communications with wireless energy
  transfer,'' \emph{IEEE Sensors Journal}.

\bibitem{dawy2015towards}
\BIBentryALTinterwordspacing
Z.~Dawy, W.~Saad, A.~Ghosh, J.~G. Andrews, and E.~Yaacoub, ``{Toward Massive
  Machine Type Cellular Communications},'' \emph{IEEE Wirel. Commun.}, vol.~24,
  no.~1, pp. 120--128, feb 2017. [Online]. Available:
  \url{http://ieeexplore.ieee.org/document/7736615/}
\BIBentrySTDinterwordspacing

\bibitem{schulz2017latency}
P.~Schulz, M.~Matthe, H.~Klessig, M.~Simsek, G.~Fettweis, J.~Ansari, S.~A.
  Ashraf, B.~Almeroth, J.~Voigt, I.~Riedel \emph{et~al.}, ``Latency critical
  iot applications in {5G}: Perspective on the design of radio interface and
  network architecture,'' \emph{IEEE Communications Magazine}, vol.~55, no.~2,
  pp. 70--78, 2017.

\bibitem{8067687}
M.~Cosovic, A.~Tsitsimelis, D.~Vukobratovic, J.~Matamoros, and C.~Anton-Haro,
  ``5g mobile cellular networks: Enabling distributed state estimation for
  smart grids,'' \emph{IEEE Communications Magazine}, vol.~55, no.~10, pp.
  62--69, OCTOBER 2017.

\bibitem{del2016d1}
J.~F.~M. del R{\'\i}o and D.~M.-S. Gand{\'\i}a, ``D1. 1 refined scenarios and
  requirements, consolidated use cases, and qualitative techno-economic
  feasibility assessment,''
  \url{https://metis-ii.5g-ppp.eu/wp-content/uploads/deliverables/METIS-II_D1.1_v1.0.pdf},
  2016.

\bibitem{1}
I.~F. Akyildiz, W.-Y. Lee, M.~C. Vuran, and S.~Mohanty, ``Next
  generation/dynamic spectrum access/cognitive radio wireless networks: A
  survey,'' \emph{Computer networks}, vol.~50, no.~13, pp. 2127--2159, 2006.

\bibitem{Nardelli2015}
P.~H. Nardelli, M.~de~Castro~Tom{\'e}, H.~Alves, C.~H. de~Lima, and
  M.~Latva-aho, ``Maximizing the link throughput between smart meters and
  aggregators as secondary users under power and outage constraints,'' \emph{Ad
  Hoc Networks}, vol.~41, pp. 57--68, 2016.

\bibitem{peha2009sharing}
J.~M. Peha, ``Sharing spectrum through spectrum policy reform and cognitive
  radio,'' \emph{Proceedings of the IEEE}, vol.~97, no.~4, pp. 708--719, 2009.

\bibitem{saleem2014primary}
Y.~Saleem and M.~H. Rehmani, ``Primary radio user activity models for cognitive
  radio networks: A survey,'' \emph{Journal of Network and Computer
  Applications}, vol.~43, pp. 1--16, 2014.

\bibitem{yu2011cognitive}
R.~Yu, Y.~Zhang, S.~Gjessing, C.~Yuen, S.~Xie, and M.~Guizani, ``Cognitive
  radio based hierarchical communications infrastructure for smart grid,''
  \emph{IEEE network}, vol.~25, no.~5, 2011.

\bibitem{gungor2012cognitive}
V.~C. Gungor and D.~Sahin, ``Cognitive radio networks for smart grid
  applications: A promising technology to overcome spectrum inefficiency,''
  \emph{IEEE Vehicular Technology Magazine}, vol.~7, no.~2, pp. 41--46, 2012.

\bibitem{op}
D.~Darsena, G.~Gelli, and F.~Verde, ``An opportunistic spectrum access scheme
  for multicarrier cognitive sensor networks,'' \emph{IEEE Sensors Journal},
  vol.~17, no.~8, pp. 2596--2606, 2017.

\bibitem{qiu2012interference}
W.~Qiu, B.~Xie, H.~Minn, and C.-C. Chong, ``Interference-controlled
  transmission schemes for cognitive radio in frequency-selective time-varying
  fading channels,'' \emph{IEEE Transactions on Wireless Communications},
  vol.~11, no.~1, pp. 142--153, 2012.

\bibitem{darsena2016convolutive}
D.~Darsena, G.~Gelli, and F.~Verde, ``Convolutive superposition for
  multicarrier cognitive radio systems,'' \emph{IEEE Journal on Selected Areas
  in Communications}, vol.~34, no.~11, pp. 2951--2967, 2016.

\bibitem{han2008cooperative}
Y.~Han, A.~Pandharipande, and S.~H. Ting, ``Cooperative spectrum sharing via
  controlled amplify-and-forward relaying,'' in \emph{IEEE 19th International
  Symposium on Personal, Indoor and Mobile Radio Communications, PIMRC
  2008.}\hskip 1em plus 0.5em minus 0.4em\relax IEEE, 2008, pp. 1--5.

\bibitem{han2009cooperative}
------, ``Cooperative decode-and-forward relaying for secondary spectrum
  access,'' \emph{IEEE Transactions on Wireless Communications}, vol.~8,
  no.~10, 2009.

\bibitem{verde2015amplify}
F.~Verde, A.~Scaglione, D.~Darsena, and G.~Gelli, ``An amplify-and-forward
  scheme for spectrum sharing in cognitive radio channels,'' \emph{IEEE
  Transactions on Wireless Communications}, vol.~14, no.~10, pp. 5629--5642,
  2015.

\bibitem{shin2011time}
E.-H. Shin and D.~Kim, ``Time and power allocation for collaborative
  primary-secondary transmission using superposition coding,'' \emph{IEEE
  Communications Letters}, vol.~15, no.~2, pp. 196--198, 2011.

\bibitem{boost}
M.~Matinmikko, M.~Latva-aho, A.~Petri, S.~Yrj\"{o}l\"{a}, and T.~Koivum\"{a}ki,
  ``Micro operators to boost local service delivery in {5G},'' \emph{Wireless
  Personal Communications}, vol. In press, May 2017.

\bibitem{Polyanskiy2010a}
Y.~Polyanskiy, H.~V. Poor, and S.~Verd{\'{u}}, ``{Channel coding rate in the
  finite blocklength regime},'' \emph{IEEE Trans. Inf. Theory}, vol.~56, no.~5,
  pp. 2307--2359, 2010.

\bibitem{chung2001design}
S.-Y. Chung, G.~D. Forney, T.~J. Richardson, and R.~Urbanke, ``On the design of
  low-density parity-check codes within 0.0045 db of the shannon limit,''
  \emph{IEEE Communications letters}, vol.~5, no.~2, pp. 58--60, 2001.

\bibitem{shannon1961two}
C.~E. Shannon \emph{et~al.}, ``Two-way communication channels,'' in
  \emph{Proceedings of the Fourth Berkeley Symposium on Mathematical Statistics
  and Probability, Volume 1: Contributions to the Theory of Statistics}.\hskip
  1em plus 0.5em minus 0.4em\relax The Regents of the University of California,
  1961.

\bibitem{ART:Durisi-PIEEE2016}
G.~Durisi, T.~Koch, and P.~Popovski, ``{Towards Massive, Ultra-Reliable, and
  Low-Latency Wireless Communication with Short Packets},'' \emph{Proceedings
  of {IEEE}}, vol. 104, no.~9, pp. 1711--1726, sep 2016.

\bibitem{Makki2014a}
B.~Makki, T.~Svensson, and M.~Zorzi, ``{Finite Block-Length Analysis of the
  Incremental Redundancy HARQ},'' \emph{IEEE Wirel. Commun. Lett.}, vol.~3,
  no.~5, pp. 529--532, oct 2014.

\bibitem{Makki2015}
------, ``{Finite Block-Length Analysis of Spectrum Sharing Networks:
  Interference-Constrained Scenario},'' \emph{IEEE Wirel. Commun. Lett.},
  vol.~4, no.~4, pp. 433--436, aug 2015.

\bibitem{Makki2016}
B.~Makki, C.~Fang, T.~Svensson, and M.~Nasiri-Kenari, ``{On the performance of
  amplifier-aware dense networks: Finite block-length analysis},'' in
  \emph{2016 International Conference on Computing, Networking and
  Communications (ICNC)}.\hskip 1em plus 0.5em minus 0.4em\relax IEEE, feb
  2016, pp. 1--5.

\bibitem{ni}
D.~Niyato, L.~Xiao, and P.~Wang, ``Machine-to-machine communications for home
  energy management system in smart grid,'' \emph{IEEE Communications
  Magazine}, vol.~49, no.~4, 2011.

\bibitem{Nardelli2014}
\BIBentryALTinterwordspacing
P.~H.~J. Nardelli, N.~Rubido, C.~Wang, M.~S. Baptista, C.~Pomalaza-Raez,
  P.~Cardieri, and M.~Latva-aho, ``Models for the modern power grid,''
  \emph{The European Physical Journal Special Topics}, vol. 223, no.~12, pp.
  2423--2437, 2014. [Online]. Available:
  \url{http://dx.doi.org/10.1140/epjst/e2014-02219-6}
\BIBentrySTDinterwordspacing

\bibitem{piti2017role}
A.~Pit{\`\i}, G.~Verticale, C.~Rottondi, A.~Capone, and L.~Lo~Schiavo, ``The
  role of smart meters in enabling real-time energy services for households:
  The italian case,'' \emph{Energies}, vol.~10, no.~2, p. 199, 2017.

\bibitem{kuzlu2014communication}
M.~Kuzlu, M.~Pipattanasomporn, and S.~Rahman, ``Communication network
  requirements for major smart grid applications in {HAN}, {NAN} and {WAN},''
  \emph{Computer Networks}, vol.~67, pp. 74--88, 2014.

\bibitem{osseiran2014scenarios}
A.~Osseiran, F.~Boccardi, V.~Braun, K.~Kusume, P.~Marsch, M.~Maternia,
  O.~Queseth, M.~Schellmann, H.~Schotten, H.~Taoka \emph{et~al.}, ``Scenarios
  for {5G} mobile and wireless communications: the vision of the metis
  project,'' \emph{IEEE Communications Magazine}, vol.~52, no.~5, pp. 26--35,
  2014.

\bibitem{fallgren2013scenarios}
M.~Fallgren, B.~Timus \emph{et~al.}, ``Scenarios, requirements and kpis for
  {5G} mobile and wireless system,'' \emph{METIS deliverable D}, vol.~1, p.~1,
  2013.

\bibitem{Durisi2016}
G.~Durisi, T.~Koch, J.~Ostman, Y.~Polyanskiy, and W.~Yang, ``{Short-Packet
  Communications over Multiple-Antenna Rayleigh-Fading Channels},'' \emph{IEEE
  Trans. Commun.}, vol.~64, no.~2, pp. 1--11, feb 2016.

\bibitem{Yang2014c}
W.~Yang, G.~Durisi, T.~Koch, and Y.~Polyanskiy, ``{Quasi-Static
  Multiple-Antenna Fading Channels at Finite Blocklength},'' \emph{IEEE Trans.
  Inf. Theory}, vol.~60, no.~7, pp. 4232--4265, jul 2014.

\bibitem{BOOK:ABRAMOWITZ-DOVER03}
M.~{Abramowitz} and I.~A. {Stegun}, \emph{Handbook of Mathematical Functions
  with Formulas, Graphs, and Mathematical Tables}, 9th~ed.\hskip 1em plus 0.5em
  minus 0.4em\relax Dover, 1965.

\bibitem{haenggi2012stochastic}
M.~Haenggi, \emph{Stochastic geometry for wireless networks}.\hskip 1em plus
  0.5em minus 0.4em\relax Cambridge University Press, 2012.

\bibitem{cardieri2010modeling}
P.~Cardieri, ``Modeling interference in wireless ad hoc networks,'' \emph{IEEE
  Communications Surveys \& Tutorials}, vol.~12, no.~4, pp. 551--572, 2010.

\bibitem{de2016contention}
C.~H. de~Lima, P.~H. Nardelli, H.~Alves, and M.~Latva-aho, ``Contention-based
  geographic forwarding strategies for wireless sensors networks,'' \emph{IEEE
  Sensors Journal}, vol.~16, no.~7, pp. 2186--2195, 2016.

\bibitem{wildman2014joint}
J.~Wildman, P.~H.~J. Nardelli, M.~Latva-aho, and S.~Weber, ``On the joint
  impact of beamwidth and orientation error on throughput in directional
  wireless poisson networks,'' \emph{IEEE Transactions on Wireless
  Communications}, vol.~13, no.~12, pp. 7072--7085, 2014.

\bibitem{elsawy2013stochastic}
H.~ElSawy, E.~Hossain, and M.~Haenggi, ``Stochastic geometry for modeling,
  analysis, and design of multi-tier and cognitive cellular wireless networks:
  A survey,'' \emph{IEEE Communications Surveys \& Tutorials}, vol.~15, no.~3,
  pp. 996--1019, 2013.

\bibitem{tome2016joint}
M.~C. Tom{\'e}, P.~H. Nardelli, H.~Alves, and M.~Latva-aho, ``Joint
  sampling-communication strategies for smart-meters to aggregator link as
  secondary users,'' in \emph{Energy Conference (ENERGYCON), 2016 IEEE
  International}.\hskip 1em plus 0.5em minus 0.4em\relax IEEE, 2016, pp. 1--6.

\bibitem{lopez2017wireless}
O.~L.~A. L{\'o}pez, E.~M.~G. Fern{\'a}ndez, R.~D. Souza, and H.~Alves,
  ``Wireless powered communications with finite battery and finite
  blocklength,'' \emph{IEEE Transactions on Communications}, 2017.

\bibitem{lin1982hybrid}
S.~Lin and P.~Yu, ``A hybrid arq scheme with parity retransmission for error
  control of satellite channels,'' \emph{IEEE Transactions on Communications},
  vol.~30, no.~7, pp. 1701--1719, 1982.

\bibitem{arq1}
P.~Larsson, L.~K. Rasmussen, and M.~Skoglund, ``Analysis of rate optimized
  throughput for large-scale {MIMO}-(h) {ARQ} schemes,'' in \emph{Global
  Communications Conference (GLOBECOM), 2014 IEEE}.\hskip 1em plus 0.5em minus
  0.4em\relax IEEE, 2014, pp. 3760--3765.

\bibitem{arq2}
T.~V. Chaitanya and E.~G. Larsson, ``Optimal power allocation for hybrid {ARQ}
  with chase combining in iid rayleigh fading channels,'' \emph{IEEE
  Transactions on Communications}, vol.~61, no.~5, pp. 1835--1846, 2013.

\bibitem{dosti2017ultra}
E.~Dosti, U.~L. Wijewardhana, H.~Alves, and M.~Latva-aho, ``Ultra reliable
  communication via optimum power allocation for type-{I} {ARQ} in finite
  block-length,'' \emph{arXiv preprint arXiv:1701.08617}, 2017.

\bibitem{el2004performance}
M.~W. El~Bahri, H.~Boujernaa, and M.~Siala, ``Performance comparison of type
  {I}, {II} and {III} hybrid {ARQ} schemes over {AWGN} channels,'' in
  \emph{Industrial Technology, 2004. IEEE ICIT'04. 2004 IEEE International
  Conference on}, vol.~3.\hskip 1em plus 0.5em minus 0.4em\relax IEEE, 2004,
  pp. 1417--1421.

\bibitem{seidel2007overview}
E.~Seidel, ``Overview {LTE} {PHY}: Part 1--principles and numerology etc,''
  \emph{Nomor 3GPP Newsletter}, 2007.

\bibitem{ts}
Nokia, ``{5G for Mission Critical Communication: Achieve ultra-reliability and
  virtual zero latency},'' \emph{Nokia White Pap.}, 2016.

\bibitem{Nguyen2015}
V.~Nguyen, Q.~Bao, L.~P. Tuyen, and H.~H. Tue, ``{A Survey on Approximations of
  One-Dimensional Gaussian Q-Function},'' 2015.

\bibitem{BOOK:GR2007}
I.~Gradshteyn and I.~Ryzhik, \emph{{Table of Integrals, Series, and Products}},
  7th~ed., A.~Jeffrey and D.~Zwillinger, Eds.\hskip 1em plus 0.5em minus
  0.4em\relax Elsevier, 2007.

\end{thebibliography}
\end{document}